\documentclass[10pt,journal,twocolumn]{IEEEtran}
\usepackage[caption=false,font=normalsize,labelfont=sf,textfont=sf]{subfig}
\IEEEoverridecommandlockouts
\usepackage{cite}
\usepackage{amsmath,amssymb,amsfonts}
\usepackage{graphicx}
\usepackage{textcomp}
\usepackage{xcolor}
\usepackage{algpseudocode}
\usepackage{algorithm,algpseudocode}
\usepackage{algorithmicx}
\usepackage{amsthm}
\usepackage{bbm}
\usepackage[normalem]{ulem}
\usepackage{multicol}

\usepackage{amsmath,epsfig,cite,amsfonts,amssymb,psfrag,color}
\usepackage{epstopdf}
\usepackage{optidef}
\usepackage{comment}

\usepackage{subfig}
\usepackage{graphicx}
\newtheorem{theorem}{Theorem}
\newtheorem{lemma}{Lemma}

\newtheorem{remark}{Remark}

\newtheorem{definition}{Definition}
\theoremstyle{plain}

\theoremstyle{plain}

\theoremstyle{plain}

\usepackage{mathtools, cuted}

\providecommand{\lemmaname}{Lemma}
\providecommand{\propositionname}{Proposition}
\usepackage{mathtools}
\usepackage{caption}
\providecommand{\theoremname}{Theorem}
\providecommand{\lemmaname}{Lemma}
\providecommand{\propositionname}{Proposition}

\def\BibTeX{{\rm B\kern-.05em{\sc i\kern-.025em b}\kern-.08em
    T\kern-.1667em\lower.7ex\hbox{E}\kern-.125emX}}
    
\def\A{\boldsymbol{\mathrm{A}}}
\def\p{\boldsymbol{\mathrm{P}}}

\def\Bset{\mathcal{B}}
\def\Nset{\mathcal{N}}
\def\Sset{\mathcal{S}}

\def\assocNBS{{a}_{b,s,n}}
\def\np{n^{\prime}}
\def\bp{b^{\prime}}

\def\pNBS{p_{b,s,n}}
\def\dNB{d_{n,b}}
\def\TxGain{G_{b}^{\text{Tx}}}
\def\RxGain{G_{n}^{\text{Rx}}}
\def\AbsCoefNBS{\exp{(-k(f_s)d_{n,b})}}

\def\Nsum{\sum\limits_{n \in \mathcal{N}}}
\def\Ssum{\sum\limits_{s \in \mathcal{S}}}
\def\Bsum{\sum\limits_{b \in \mathcal{B}}}

\def\Npsum{\sum\limits_{n^{\prime} \in \mathcal{N}^{\prime}}}

\def\Bpsum{\sum\limits_{b^{\prime} \in \mathcal{B}^{\prime}}}

\def\h{{\lvert{h_{b,s,n}}\rvert}^2}
\def\hbar{{\lvert{\bar{h}_{b,s,n}}\rvert}^2}
\def\hbpbar{{\lvert{\bar{h}_{b^{\prime},s,n}}\rvert}^2}
\def\htilde{{\lvert{\tilde{h}_{b,s,n}}\rvert}^2}

\def\gam{\gamma_{b,s,n}}

\def\pr{\mathcal{P}}
\def\pr{\mathcal{P}}
\newcommand{\abss}[1]{{\left\lvert{#1}\right\rvert}^2}
\usepackage{pifont}
\newcommand{\cmark}{\ding{51}}%
\newcommand{\xmark}{\ding{55}}%

\begin{document}
\raggedbottom

\title{Molecular Absorption-Aware User Assignment, Spectrum, and Power Allocation  in  Dense THz Networks with Multi-Connectivity}

\author{Mohammad Amin Saeidi, {\em Graduate Student Member IEEE},  Hina~Tabassum, {\em Senior Member IEEE}, and Mehrazin Alizadeh, {\em Student Member IEEE}
\thanks{This research is supported by a Discovery Grant funded by the Natural Sciences and Engineering Research Council of Canada. The authors are with the Department of Electrical Engineering and Computer Science at York University, Toronto, ON M3J 1P3, Canada. (e-mail: \{amin96a, hinat, Ma75\}@yorku.ca ). }% <-this % stops a space
\vspace{-10mm}
%\thanks{This research was supported by a Discovery Grant funded by the Natural Sciences and Engineering Research Council of Canada.}
% \\
%  Department of Electrical Engineering and Computer Science, York University, Toronto, ON, Canada
% \\
% Email: amin96a@yorku.ca, hinat@yorku.ca, mehrazin@yorku.ca
}

\maketitle
\raggedbottom
\begin{abstract}
% The Terahertz (THz) spectrum is emerging as a potential candidate for 6G wireless networks as it offers extensive bandwidth to enable extreme data rates. However, THz transmissions suffers from significant propagation losses due to  molecular absorption in the environment, molecular absorption noise, beam-squint and blockages.
This paper develops a unified framework to maximize the network sum-rate in a multi-user, multi-BS downlink terahertz (THz) network by optimizing user associations, number and bandwidth of sub-bands in a THz transmission window (TW), bandwidth of leading and trailing edge-bands in a TW,  sub-band assignment, and power allocations. The proposed framework incorporates multi-connectivity and captures the impact of molecular absorption coefficient variations in a TW, beam-squint, molecular absorption noise, and link blockages. To make the problem tractable, we first propose a convex approximation of the molecular absorption coefficient using curve fitting in a TW, determine the feasible bandwidths of the leading and trailing edge-bands, and then derive closed-form optimal solution for the number of sub-bands considering beam-squint constraints. We then  decompose joint user associations,  sub-band assignment, and power allocation problem into two sub-problems, i.e., \textbf{(i)} joint user association and sub-band assignment, and \textbf{(ii)} power allocation. To solve the former problem, we analytically prove the unimodularity of the constraint matrix which enables us to relax the integer constraint without loss of optimality. To solve power allocation sub-problem, a  fractional programming (FP)-based centralized solution as well as an alternating direction method of multipliers (ADMM)-based light-weight distributed solution is proposed. The overall problem is then solved using alternating optimization until convergence. Complexity analysis of the algorithms and numerical convergence are presented. Numerical findings validate the effectiveness of the proposed algorithms and extract useful insights about the interplay of the density of base stations (BSs), Average order of multi-connectivity (AOM), molecular absorption, {hardware impairment}, {imperfect CSI}, and link blockages.
\end{abstract}

\begin{IEEEkeywords}
Terahertz communication, joint user association and sub-band assignment, unimodularity, distributed power allocation, blockage, molecular absorption, hardware impairment
\end{IEEEkeywords}

\section{Introduction}
The evolution of wireless cellular systems, known as the sixth generation (6G), will be driven by the need to support revolutionary services such as connected autonomous vehicles (CAVs), extended reality (XR), digital twins, and immersive remote presence. To meet the requirements of these bandwidth-hungry services, 6G wireless systems must deliver unprecedented data rates \cite{6G-Trends,saeidi2023tractable,THz-URLLC,R2-REF-1,MBN-Survey}. Terahertz (THz) spectrum (0.1 THz to 10 THz) has the potential to offer much wider transmission bandwidths with extreme data rates (in the order of multi-Gbps). Nonetheless, THz spectrum comes with its own set of challenges. Different from the traditional radio frequencies ($<$ 6GHz), THz transmissions are susceptible to high propagation losses due to molecular absorption in the environment, molecular absorption noise, blockages, and hardware impairments. The ultra-broad bandwidth contributes to non-negligible thermal noise.  The THz transmissions are also vulnerable to beam squint which causes beams at various frequencies to radiate in different physical directions leading to system performance degradation \cite{saeidi2023multiband,Beam-squint-DL-based,R2-REF-2}.  As a result, dense THz network deployments will be critical to overcome the molecular absorption, blockage, and noise effects. 

To enable large-scale THz  communication, several critical factors must be considered, including low complexity, efficient and optimal power allocation, user association, and sub-band assignment. It is also essential to account for interference, molecular absorption noise resulting from molecular absorption phenomena at high frequencies, and blockage, while also addressing the beam squint phenomenon.
{Furthermore, molecular absorption is more pronounced in the THz band compared to the mmWave band. Due to the numerous absorption peaks in the THz band, we need to consider transmission windows (TWs), which are unnecessary in the mmWave band. The high molecular absorption at different frequencies in the THz band makes the channel, and achievable rates frequency-dependent, adding challenges to resource allocation. Therefore, we must take specific steps, such as managing TWs, determining trailing and edge bands to maintain constant molecular absorption, and deciding the total number of sub-bands to avoid beam squint phenomena in the THz band.}

In light of these challenges, this work presents a novel and low-complexity framework with multi-connectivity that caters to both centralized and distributed power allocation, while tackling the joint user association and sub-band assignment problem and carefully considering blockage and beam squint effects. {This work incorporates multiple constraints into the optimization problem to determine the parameters of a given TW and the optimal number of sub-bands, aiming to enhance the overall system sum-rate.} 

\subsection{Related Works}
A handful of research works investigated the \textit{spectrum and power allocation in single-cell THz} wireless networks without interference. In \cite{Distance-Aware-2016}, a distance-aware and bandwidth-adaptive resource allocation was proposed to maximize the communication distance. The system considered a single BS and interference from adjacent sub-bands is modeled as Gaussian noise. Specifically, sub-band assignment using distance-ranking and power control using convex optimization were proposed. In \cite{UAV-BW-Power}, the minimization of total uplink and downlink delay in a UAV-based THz transmission was considered. The authors  optimized bandwidth, users' power allocation, and the UAV's location by decomposing the joint problem into three convex sub-problems and solved them via alternating optimization.  
The work in \cite{Tera-IoT} studies how fast and far bits can be transmitted in the THz band by maximizing the product of rate and distance. The paper first obtains the sub-window allocation utilizing the Hungarian algorithm and then alternatively optimizes power allocation and transmission distance.
% In \cite{IRS-THz-1},  sub-band assignment, power allocation, and IRS phase and location optimization are conducted to maximize the system sum-rate with alternating optimization. 
Moreover, the authors in \cite{shafie2022novel} considered sum-rate maximization in a single-cell THz downlink system by optimizing power allocation, sub-band assignment as well as their respective bandwidth allocation. The authors approximated the molecular absorption coefficient as an exponential function of frequency and relax integer constraints by adding a non-convex constraint. Then, they applied successive convex approximation (SCA) to solve the problem.  
Following a similar approach in \cite{shafie2022novel}, the authors considered maximizing the minimum rate in the multi-cell uplink THz network by optimizing power allocation, user association, sub-band assignment as well as their respective bandwidth allocation \cite{shafie2021spectrum}. 
Nonetheless, a sub-band is assigned to only one user at a time and the total number of sub-bands is greater than or equal to the number of users. Thus, there is no interference from other users or BSs.

None of the aforementioned research considered resource allocation in large-scale THz networks considering interference, blockages, molecular absorption noise, and/or multi-connectivity. 
% Also, the adaptive bandwidth allocation proposed in most works makes interference formulation and computations challenging in a multi-cell set-up with frequency reuse.
{Specifically, the adaptive bandwidth allocation used in \cite{shafie2021spectrum} and \cite{shafie2022novel} makes the interference formulation challenging in multi-cell systems when considering frequency reuse. In the case of variable bandwidth, interference can be non-uniform due to non-uniform bandwidth divisions across different BSs and varying transmit power on those bandwidth divisions. Moreover,  adaptive bandwidth allocation increases the complexity of resource allocation, as in \cite{shafie2022novel} for even a single-cell set-up. Additionally, the impact of molecular absorption noise was not considered. 
}
% In addition, the so-called molecular absorption noise that occurs in the denominator of the signal-to-interference-and-noise ratio (SINR) is neglected in \cite{shafie2021spectrum,shafie2022novel}. 

Recently, a handful of research works touched on the resource allocation in large-scale THz networks though without considering the propagation specifics of THz channel. The work in \cite{EE-THz-NOMA} maximizes energy efficiency (EE) in a downlink THz system. Matching theory (switch-matching) is used for sub-channel assignment and power allocation is performed using the alternative direction method of multipliers (ADMM). In \cite{EE-Resource-Allocation-Cache-Based}, sub-bands and power allocation are optimized to maximize EE in cache-based THz vehicular networks. This work, without taking TWs and molecular absorption noise into account, first obtains a simplified expression for the rate and then relaxes the binary variables for sub-band assignment. The proposed solution is based on Dinkelbach and Lagrangian dual methods. Authors in \cite{moldovan2017coverage} investigate multiple single antenna APs transmitting the signal of all users on the same frequency. In terms of molecular absorption windows, the maximum value of molecular absorption noise is considered, which leads to a lower bound on the performance. This paper, which aims to maximize the minimum rates, decomposes the power allocation and sub-band assignment problems and obtains a heuristic solution for the latter, and utilizes convex optimization for the former. {Table I summarizes the existing works on the THz resource allocation and clarifies the novelties of this work.}  

\begin{table*}[]
\caption{Comparison of works on THz resource allocation}
\label{tab:RefComp}
\centering
\scalebox{0.84}{
{\color{black}
\begin{tabular}{|c|c|c|c|c|c|c|c|c|c|c|c|}
\hline
\textbf{Ref.}                             & \textbf{\begin{tabular}[c]{@{}c@{}}Frequency \\ reuse\end{tabular}} & \textbf{\begin{tabular}[c]{@{}c@{}}User \\ assignment\end{tabular}} & \textbf{\begin{tabular}[c]{@{}c@{}}Multi-\\ Connectivity\end{tabular}} & \textbf{\begin{tabular}[c]{@{}c@{}}Considering\\ TWs\end{tabular}} & \textbf{\begin{tabular}[c]{@{}c@{}}Molecular \\ absorption \\ noise\end{tabular}} & \textbf{\begin{tabular}[c]{@{}c@{}}Beam \\ squint\end{tabular}} & \textbf{Method}                                                        & \textbf{\begin{tabular}[c]{@{}c@{}}Computational\\ Complexity\end{tabular}} & \textbf{\begin{tabular}[c]{@{}c@{}}Hardware\\  Impairment\end{tabular}} & \textbf{Blockage} & \textbf{\begin{tabular}[c]{@{}c@{}}Imperfect\\  CSI\end{tabular}} \\ \hline
\cite{UAV-BW-Power}                       & \xmark                                                              & \xmark                                                              & \xmark                                                                 & \xmark                                                             & \xmark                                                                            & \xmark                                                          & Alternating Opt.                                                       & Low                                                                         & \xmark                                                                  & \xmark            & \xmark                                                            \\ \hline
\cite{Tera-IoT}                           & \xmark                                                              & \xmark                                                              & \xmark                                                                 & \cmark                                                             & \xmark                                                                            & \xmark                                                          & \begin{tabular}[c]{@{}c@{}}Alternating Opt.\\ + Hungarian\end{tabular} & High                                                                        & \xmark                                                                  & \xmark            & \xmark                                                            \\ \hline
\cite{shafie2022novel}                    & \xmark                                                              & \xmark                                                              & \xmark                                                                 & \cmark                                                             & \xmark                                                                            & \xmark                                                          & SCA                                                                    & High                                                                        & \xmark                                                                  & \xmark            & \xmark                                                            \\ \hline
\cite{shafie2021spectrum}                 & \xmark                                                              & \cmark                                                              & \cmark                                                                 & \cmark                                                             & \xmark                                                                            & \xmark                                                          & SCA                                                                    & High                                                                        & \xmark                                                                  & \cmark            & \xmark                                                            \\ \hline
\cite{EE-THz-NOMA}                        & \cmark                                                              & \xmark                                                              & \xmark                                                                 & \xmark                                                             & \xmark                                                                            & \xmark                                                          & \begin{tabular}[c]{@{}c@{}}Matching Theory \\ + ADMM\end{tabular}      & High                                                                        & \xmark                                                                  & \xmark            & \xmark                                                            \\ \hline
\cite{EE-Resource-Allocation-Cache-Based} & \cmark                                                              & \xmark                                                              & \xmark                                                                 & \xmark                                                             & \xmark                                                                            & \xmark                                                          & \begin{tabular}[c]{@{}c@{}}Mean-Field \\ Game\end{tabular}             & High                                                                        & \xmark                                                                  & \xmark            & \xmark                                                            \\ \hline
\cite{moldovan2017coverage}               & \xmark                                                              & \xmark                                                              & \xmark                                                                 & \xmark                                                             & \cmark                                                                            & \xmark                                                          & Heuristic                                                              & High                                                                        & \xmark                                                                  & \xmark            & \xmark                                                            \\ \hline
This paper                                & \cmark                                                              & \cmark                                                              & \cmark                                                                 & \cmark                                                             & \cmark                                                                            & \cmark                                                          & \begin{tabular}[c]{@{}c@{}}Alternating Opt.\\ + FP/ADMM\end{tabular}   & Low                                                                         & \cmark                                                                  & \cmark            & \cmark                                                            \\ \hline
\end{tabular}
}
}
\end{table*}

\subsection{ Contributions}
% Majority of the aforementioned research works focused on adaptive THz bandwidth assignment to combat molecular absorption variations. This bandwidth adaptation makes it challenging to extend the resource allocation solutions in large-scale dense THz networks with molecular absorption noise, interference, and blockages. 

To our best knowledge, none of the aforementioned research works addressed the problem of power allocation, joint user association, and sub-band partition and assignment in a large-scale THz network for a given THz TW, while accounting for multi-connectivity, blockage-aware rate, molecular absorption noise, efficient spectrum reuse, and beam squint. 
The main contributions of this work are summarized as follows:

 $\bullet$ We maximize the overall network sum-rate in a downlink THz network with multiple users and multiple BSs by strategically optimizing various parameters. These include user assignments, the count and width of sub-bands within a THz TW, the width of leading and trailing edge-bands within the TW, sub-band, and power allocations.  Our approach considers multi-connectivity, factors in molecular absorption coefficient variations, beam-squint, molecular absorption noise, and link blockages. For beam-squint, we employ the fractional bandwidth concept to set an upper limit for each sub-band bandwidth. {In particular, compared to the existing works, making uniform partitions of sub-bands over a given TW allows us to concurrently enable frequency reuse by including the interference terms and determine an optimal number of sub-bands in favor of achieving a higher system sum-rate and avoiding beam squint. Additionally, incorporating interference renders the problem highly non-convex, necessitating the development of low-complexity algorithms for solutions.}
     
$\bullet$ To determine the number of sub-bands required, we initially employ curve-fitting techniques to create an analytical function representing the TW of interest. This fitted curve allows us to derive the values of the TW's edge bands, which are the regions where transmission should be avoided. By using the estimated edge band values alongside the beam squint-aware upper bound on the bandwidth, we obtain a lower bound on the total number of sub-bands needed for the system.
    
$\bullet$ With the number of sub-bands established, we divide the optimization problem into two sub-problems: joint user association and sub-band assignment, and power allocation. To address the non-convex nature of the association and sub-band assignment sub-problem, we reduce the dimension of the power allocation variables from 3-dimensional (3D) to 2D. Leveraging the structure of the constraints and employing the concept of total unimodularity, we transform the {nondeterministic polynomial time} (NP-hard) integer programming into a low-complexity linear programming, offering a novel and comprehensive proof in this context.
    
$\bullet$ We address the power allocation sub-problem through two proposed fractional programming (FP)-based approaches. One approach involves centralized FP, while the other introduces a novel and low-complexity scheme based on the ADMM methodology. This distributed scheme reduces signaling and communication overhead significantly. Employing alternating optimization, the overall low-complexity algorithm iteratively solves the joint user association and sub-band assignment, as well as the power allocation sub-problem.

$\bullet$ {Numerical results indicate that leveraging multi-connectivity presents performance improvement compared to single-connectivity.  Additionally, when evaluating the impact of antennas' directionality, blockage, molecular absorption coefficient, {hardware impairment}, {imperfect CSI}, and the number of BSs, the proposed methods outperform the traditional benchmarks.} {Our proposed framework offers a general solution that can also be applied to lower frequencies. By setting the molecular absorption coefficient to zero and removing TWs, and by using the respective channel model and adjusting the formulations, our framework can be adapted for lower frequencies such as mmWave and sub-6GHz.}

\textbf{Notations:} This paper uses boldface lowercase letters to represent vectors and boldface capital letters to represent matrices. ${\{0,1\}}^{B\times N}$ and $\mathbb{R}_+^{B\times N}$ denote the space of $B\times N$ binary-valued and non-negative real-valued matrices, respectively. $\mathbf{I}_{N}$ stands for the $N\times N$ identity matrix. $\lvert x \rvert$ is used for the absolute value of $x$, ${\lVert \mathbf{x} \rVert}_2$ denotes the L2 norm of vector $\mathbf{x}$, and $\mathbf{X}^T$ is the transpose of matrix $\mathbf{X}$. $\text{diag}(\mathbf{X},B)$ is a block diagonal matrix, such that $\mathbf{X}$ is repeated $B$ times diagonally, and $\otimes$ indicates Kronecker product. $\mathcal{CN}(0,\sigma^2)$ denotes complex Gaussian distribution with zero mean and variance of $\sigma^2$.

\section{System Model and Assumptions}

\begin{figure}
   \centering
\includegraphics[scale=0.44]{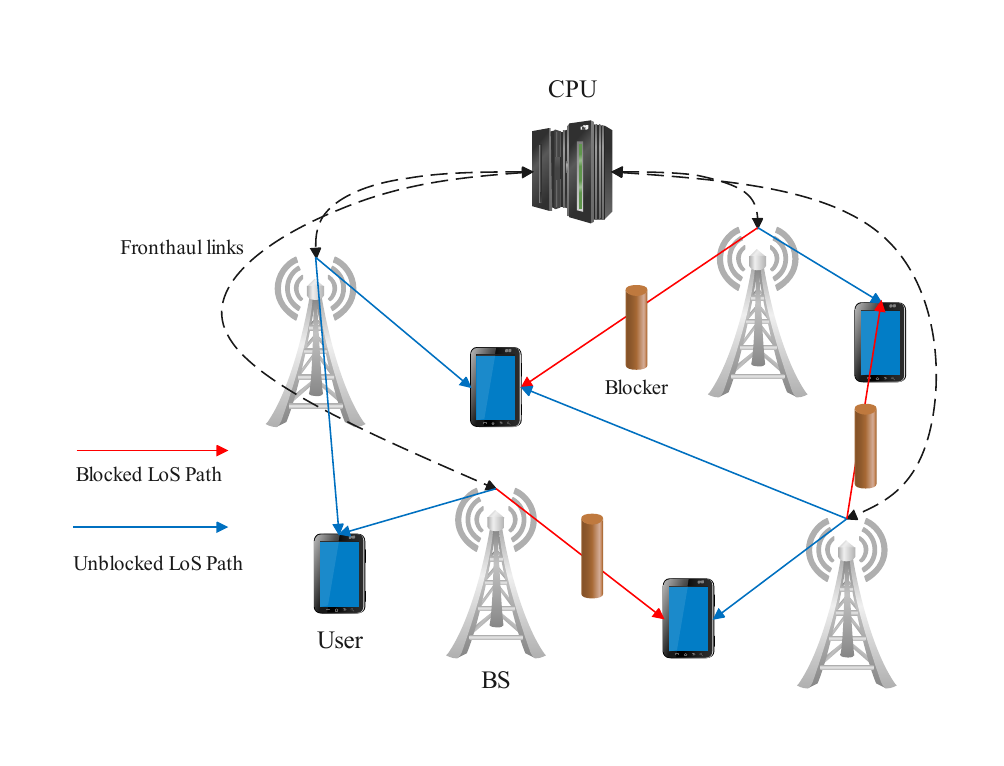}
\caption{{System model illustrating BSs and users in the presence of randomly distributed blockers.}}
\label{fig:SysModel}
\end{figure}

\subsection{Downlink Transmission Model}
{As shown in Fig.~\ref{fig:SysModel},} we consider a downlink THz transmission in a multi-user multi-cell system, where $N$ users are served by $B$ BSs. We define $\Bset=\{1,2,\cdots,b,\cdots,B\}$ as the set of all THz BSs deployed in a cellular region and the set of all users in the cellular region is denoted by $\Nset=\{1,2,\cdots,n,\cdots,N\}$.  

Each of the THz BS and the user is equipped with a single directional antenna.
% Note that the considered scenario is not limited to indoor communication.
Due to the existence of path-loss peaks  in the THz spectrum, we consider downlink communication in specified transmission windows where the impact of molecular absorption is minimum. We consider dividing a given TW into $S$ orthogonal sub-bands which is determined based on beam-squint and the bandwidth of edge-bands (as will be detailed later). Additionally, guard bands are utilized between each sub-band to prevent inter-band interference.
The set of sub-bands is referred to as $\Sset=\{1,2,\dots,s,\dots,S\}$.
Each sub-band is reused by all BSs, hence intra-band interference exists.

Without loss of generality, we consider $f_1 < f_2 < \dots < f_S$. Thus, the central frequency of $s$-th sub-band is given as
$f_s = f_I + w_I + (s-0.5)w + (s-1)w_G,$
where $w$ is the bandwidth of each sub-band, $f_I$ is the initial carrier frequency of the TW of interest, $w_I$ is the edge bandwidth at the beginning of the TW, and $w_G$ is the guard band between two consecutive sub-bands. The total bandwidth of a given TW can thus be given as, $w_T=f_E - f_I$, where $f_E$ is the final frequency of the TW. Thus, spectrum allocation is performed in the spectral range of $ [f_I + w_I, \quad  f_E - w_E]$, where $w_E$ is the edge bandwidth at the end of the given TW. 
Both $w_I$ and $w_E$ are designed to avoid transmissions at frequencies with significantly varying molecular absorption. 
% Then the following constraint should be regarded in the resource allocation.
% \begin{equation}\label{eqn:edgeBand}
%     f_I + w_I \leq Sw + (S-1)w_G \leq f_E - w_E
% \end{equation}
We assume multi-connectivity, i.e., each user can be associated with multiple BSs on a different sub-band. Moreover, each BS assigns a sub-band to one user to avoid intra-band within each cell.  

\subsection{Beam Squint Consideration and THz Channel Model}
The ultra-broadband nature of THz transmission can simply increase the ratio between the central frequency, $f_s$, and its bandwidth, $w$. This ratio, called fractional bandwidth, can lead to the beam squint effect, in which each frequency component is divided into various directions leading to beam gain degradation \cite{TeraMINO}. In order to reduce the impact of the beam squint effect, we consider limiting the fractional bandwidth, i.e., $B_s$ of the $s$-th sub-band within $B_{\mathrm{th}}$.
\begin{equation}\label{eqn:BeamSquint1}
    B_s = \frac{w}{f_s}=\frac{w}{f_I + w_I + (s-0.5)w + (s-1)w_G} \leq B_{th}.
\end{equation}
% By rewriting \eqref{eqn:BeamSquint1}, we have
% \begin{equation}\label{eqn:BeamSquint2}
%     w \leq \frac{f_I B^{Th} + w_I B^{Th} + (s-1)B^{Th}}{1-(0.5-s)B^{Th}}, \forall s
% \end{equation}
% It is worth noting that, avoiding beam squint is required for the system under consideration since each user can be assigned to more than one sub-band, and all assigned sub-bands must be transmitted in the same physical direction based on the user's position. Hence, by employing \eqref{eqn:BeamSquint1}, beam squint-induced system performance degradation can be mitigated in such a system.

The channel power of the THz link from $b$-th BS to $n$-th user on sub-band $s$ is modeled as:
\begin{equation}\label{eqn:channel}
    \h = \TxGain\RxGain c^2 \dNB^{-2}\AbsCoefNBS \chi_{b,s,n}{(4\pi f_s)}^{-2},
\end{equation}
where  $\TxGain$ and $\RxGain$ are the antenna gains of $b$-th BS and $n$-th user, respectively \cite{THz-Channel-1}, and $\dNB$ is the distance between $b$-th BS and $n$-th user. Also, $c$ is the speed of light, $k(f_s)$ is molecular absorption coefficient at $f_s$, and $N_0$ is the additive Gaussian noise spectral density, and $\chi_{n,b,s}$ denotes Nakagami-$m$\footnote{In the THz spectrum, the non-line of sight (NLoS) paths resulting from scattering and refraction are negligible \cite{EE-Resource-Allocation-Cache-Based}. {Also, the Nakagami-$m$ fading channel model is a general model that can approximate Rician fading for large values of shape parameter ($m\gg1$).}}  fading channel to model line-of-sight (LoS) transmissions \cite{THz-Channel-1}. \footnote{{All channels are assumed to be estimated during the channel training phase and sent to the central processing unit (CPU) via fronthaul links. However, the impact of imperfect channel state information (CSI) on system performance is discussed in Section VI.}}
The antenna gain for all BSs and users is a function of the angle from the antenna boresight direction, $\theta^g \in [-\pi,\pi)$, and main lobe beamwidth, $\Theta^g$, where $g \in \{\text{Tx},\text{Rx}\}$ \cite{AntennaGain}. Thus, the antenna gain can be modeled as follows: 
\begin{equation}\label{eqn:AntennaGain}
G^g(\theta^g)= \left\{ \begin{array}{ll}
G^g_{\max}, &  |\theta^g| \leq \Theta^g \\
G^g_{\min} , & |\theta^g| > \Theta^g,
\end{array}\right.
\end{equation}
where $G^g_{\max}$ and $G^g_{\min}$ are the array gains of the main and side lobes, respectively. When the main lobes of a user and its associated BS coincides, the antenna gain of that link becomes $G^{\text{Tx}}_{\max}G^{\text{Rx}}_{\max}$. Similar to \cite{TanvirMobilityAware}, the alignment probability between the main lobes of interfering BSs and a specific user is denoted by $q=\frac{\theta^{\text{Tx}}\theta^{\text{Rx}}}{4\pi^2}$, and the side-lobes are negligible, i.e., $G_{\min}^g\approx 0$. As the antenna's directionality increases, the main lobe's beamwidth proportionally decreases. Consequently, based on \eqref{eqn:AntennaGain}, when the antenna exhibits higher directionality (characterized by a smaller $\Theta^g$ value), $q$ diminishes because both $\theta^{\text{Tx}}$ and $\theta^{\text{Rx}}$ are limited to values smaller than $\Theta^g$. Thus, we consider $q$ as a factor denoting the directionality of the antennas at both users and BSs. Specifically, if $q \approx 0$, it indicates that the interfering BSs have high directionality and that inter-band interference can be disregarded, whereas if $q=1$, the user captures the entire transmit power of the interfering BSs.

\subsection{THz SINR Model and Blockage-Aware Rate Model}
Let us define the joint user association and sub-band assignment variable as:
\begin{equation}\label{eqn:AssociationVar}
a_{b,s,n} = \!\left\{\begin{array}{ll}
\!\!\!1, &  \!\!\!\text{if user $n$ is associated with BS $b$ on sub-band $s$} \\
\!\!\!0, & \text { otherwise }.
\end{array}\right.
\end{equation}
Then, the matrix of joint user association and sub-band assignment is given as $\A \in \{0,1\}^{B S \times N}$, such that $a_{b,s,n}$ is the entry of $\A$ at $((b-1)S+s)$-th row and $n$-th column, {and $BS$ means the multiplication of $B$ and $S$.} 

% The SINR at $n$-th user associated with $b$-th BS on sub-band $s$ can be given as\cite{TanvirMobilityAware}:
% \begin{equation}
%     \gamma_{b,s,n} (\A,\p,w) \! = \! \frac{a_{b,s,n}\pNBS \h }{I_{b,s,n}(\p,\A) + \Phi_{b,s,n}(\p,\A)  + N_0 w},
% \end{equation}
% where $\pNBS$ is the transmit power from $b$-th BS to $n$-th user on sub-band $s$. Thus, the matrix representation of the power allocation can be given as ${\p}=\{p_{n,b,s}\} \in \mathbb{R}_+^{B S\times N}, \ \forall b \in \Bset, \ n \in \Nset, s \in \Sset$. Spectral reuse in all BSs leads to intra-band interference modeled as  follows:
% \begin{equation}\label{eqn:Intra-band-interfernce}
%     I_{b,s,n}(\p,\A) = \Bpsum \Npsum q a_{\bp,s,\np} p_{\bp,s,\np} \abss{h_{\bp,s,n}},
% \end{equation}
% such that $\Nset^{\prime} \triangleq \Nset \backslash\{n\}$, $\Bset^{\prime} \triangleq \Bset \backslash\{b\}$.
{
We define the matrix representation of the power allocation as ${\p}=\{p_{b,s,n}\} \in \mathbb{R}_+^{B S\times N}, \ \forall b \in \Bset, \ n \in \Nset, s \in \Sset$ where $\pNBS$ is the transmit power from $b$-th BS to $n$-th user on sub-band $s$. Then, the spectral reuse in all BSs, which leads to intra-band interference, is modeled as follows:
\begin{equation}\label{eqn:Intra-band-interfernce}
    I_{b,s,n}(\p,\A) = \Bpsum \Npsum q a_{\bp,s,\np} p_{\bp,s,\np} \abss{h_{\bp,s,n}},
\end{equation}
such that $\Nset^{\prime} \triangleq \Nset \backslash\{n\}$, $\Bset^{\prime} \triangleq \Bset \backslash\{b\}$.
}

In addition to thermal noise, it is essential to consider molecular absorption noise, which is induced by the re-radiation of absorbed signal energy, in the THz band. In particular, as a result of using superconductive materials in the THz frequency, thermal noise can be decreased, therefore molecular absorption noise would be comparable to thermal noise and must be taken into account \cite{Absorption-Noise}. Molecular absorption noise for user $n$ associated to BS $b$ over sub-band $s$ is modeled as follows:
\begin{equation*}
    \hspace{-3cm}\Phi_{b,s,n}(\A,\p) = \assocNBS \pNBS {\lvert{\tilde{h}_{b,s,n}}\rvert}^2
\end{equation*}
\begin{equation}\label{eqn:MolecularNosie}
     + \Bpsum \Npsum q a_{\bp,s,\np} p_{\bp,s,\np} {\lvert{\tilde{h}_{\bp,s,n}}\rvert}^2,
\end{equation}
where the first term is related to the molecular noise of the desired link and the second term is related to the molecular noise of the interfering links. Also, the channel gain of the molecular absorption noise is given as $\htilde = C \dNB^{-2} (1-\exp(-k(f_s)d_{n,b})) \chi_{b,s,n} {f_s}^{-2}$, such that $ C = \TxGain\RxGain c^2{(4\pi)}^{-2} \ \forall b,n$. Hence, the cumulative interference and molecular absorption noise from all interfering BSs, {which is the summation of $I_{b,s,n}(\p,\A)$ and the second term in \eqref{eqn:MolecularNosie}}, is obtained as:
\begin{equation}\label{eqn:CumulInter}
    \bar{I}_{b,s,n}(\p,\A)=\Bpsum \Npsum q a_{\bp,s,\np}p_{\bp,s,\np}\abss{\bar{h}_{\bp,s,n}},
\end{equation}
where $\hbar = C \dNB^{-2} \chi_{b,s,n} {f_s}^{-2}$.
% \begin{equation}\label{eqn:InterferingMolecularNoise}
%     \begin{split}
%          \Phi_{b,s,n}(\A,\p) = \Bpsum \Npsum  
%          q C  \A_{\bp,s,\np} p_{\bp,s,\np} d^{-2}_{n,\bp}(1-\exp{(-k(f_s)d_{n,\bp}))} \chi_{\bp,s,\np} f_s^{-2}. 
%     \end{split}
% \end{equation}
Then, the SINR at $n$-th user associated with $b$-th BS on sub-band $s$ can be given as\cite{TanvirMobilityAware}:
\begin{equation}\label{eqn:SINR}
    \gamma_{b,s,n} (\A,\p,w) 
     = \frac{\assocNBS\pNBS \h}{\bar{I}_{b,s,n}(\p,\A) \! + \! \assocNBS \pNBS {\lvert{\tilde{h}_{b,s,n}}\rvert}^2 \!\!\!+\!\! N_0 w}.
\end{equation}
Therefore, the achievable rate at $n$-th user associated with $b$-th BS over sub-channel $s$ is stated as follows:
\begin{equation}\label{eqn:Rate-b-s}
 \bar{R}_{b,s,n}(\A,\p,w) = w \log_2 (1+\gamma_{b,s,n}(\A,\p,w)).
\end{equation}
THz transmission is susceptible to blockage due to high penetration losses \cite{saeidi2023multiband}; thus, we incorporate link blockage probability. We define $\boldsymbol{\Psi} \in {\{0,1\}}^{B\times N}$ as the indicator variable for blockage such that $\boldsymbol{\Psi}_{b,n} = 1$ if there is a blockage between user $n$ and BS $b$, and otherwise $\boldsymbol{\Psi}_{b,n} = 0$. Each element of $\boldsymbol{\Psi}$ follows the following distribution \cite{Blockage-aware1,AntennaGain}:
\begin{equation}\label{eqn:Blockage-Dist}
\boldsymbol{\Psi}_{b,n}= \left\{ \begin{array}{ll}
1, &  \text{With probability} \ \ 1-e^{-\eta d_{b,n}} \\ 0, & \text{With probability} \ \ e^{-\eta d_{b,n}},
\end{array}\right.
\end{equation}
where $\eta$ is the density of blockers in a given area. For a specific $\eta$, the likelihood of blocking a link between a user and BS increases exponentially with their distance. 
The random blockage matrix $\boldsymbol{\Psi}$ is multiplied by the channel between a user and BS. The blockage-aware rate between a user and BS is thus modeled as follows: If $\boldsymbol{\Psi}_{b,n} = 0$, the rate is $R_{b,s,n}(\A,\p,w,\boldsymbol{\Psi})=\bar{R}_{b,s,n} (\A,\p,w)$, otherwise $R_{b,s,n}(\A,\p,w,\boldsymbol{\Psi})=0$.

\section{Problem Formulation, Constraints, and Transformations}
\subsection{Problem Formulation}
The  sum-rate maximization problem in a multi-user, multi-BS downlink THz network is formulated in the following:
% to optimize  user association, number and bandwidth of sub-bands, bandwidth of edge bands,  sub-band assignment, and power allocations. 
% \begin{maxi!}[3]
% 	{\A,\p,w,w_E,w_I,S}{\Bsum \Nsum \Ssum \omega_{n,b} R_{b,s,n} \label{Con:p1Obj}}
% 	{\label{problem:p1main}}{\mathcal{P}_1: \ \ }
% 	\addConstraint{\Nsum \Ssum \assocNBS \pNBS}{\leq P_b^{\max}, \forall b \in \Bset}\label{Con:p1Power}
%     \addConstraint{\Bsum \assocNBS \leq 1, \forall s \in \Sset, n \in \Nset}{}\label{Con:p1Assoc1}
%     \addConstraint{\Nsum \assocNBS = 1, \forall b \in \Bset, s \in \Sset}{}\label{Con:p1Assoc2}
%     \addConstraint{\Gamma_{n}^L \leq \Bsum \Ssum \assocNBS \leq \Gamma_{n}^U}{, \forall n \in \Nset}\label{Con:p1Fairness}
%     \addConstraint{w_I + Sw + (S-1)w_G + w_E}{= w_{T}}\label{Con:p1WTotal}
%     \addConstraint{f_I + w_I \leq Sw + (S-1)w_G \leq f_E - w_E}{}\label{Con:p1TW-bounds}
%     \addConstraint{w \leq \frac{f_I B^{Th} + w_I B^{Th} + (s-1)B^{Th}}{1-(0.5-s)B^{Th}}, \forall s \in \Sset}\label{Con:p1FracBW}
%     \addConstraint{|k(f_E-w_E)-k(f_I+w_I)|\leq \epsilon}{}\label{Con:p1MolCoeff}
%     \addConstraint{\assocNBS \in \{0,1\},}{\forall n\in\Nset, b \in \Bset, s \in \Sset}\label{Con:p1Integer}
%     \addConstraint{S \in \mathbb{N}}{}\label{Con:p1Sinteger}.
% \end{maxi!}
\begin{eqnarray}
\label{prob1} \pr_1:\hspace*{-4mm}
&&\hspace*{2mm}\underset{\A,\p,w,w_E,w_I,S}{\textrm{maximize}} \,\, \,\, \Bsum \Nsum \Ssum R_{b,s,n}\\
\mbox{s.t.}\hspace*{-3mm}
&&\hspace*{-2mm}\mathbf{C_1}: \Nsum \Ssum \assocNBS \pNBS \leq P_b^{\max}, \forall b \in \Bset,\notag\\
&&\hspace*{-2mm}\mathbf{C_2}: \Bsum \assocNBS \leq 1, \forall s \in \Sset, n \in \Nset,\notag\\
&&\hspace*{-2mm} \mathbf{C_3}:\Nsum \assocNBS = 1, \forall b \in \Bset, s \in \Sset,\notag\\
&&\hspace*{-2mm}\mathbf{C_4}: \Gamma_{n}^L \leq \Bsum \Ssum \assocNBS \leq S, \forall n \in \Nset,\notag\\
&&\hspace*{-2mm} \mathbf{C_5}: w_I + Sw + (S-1)w_G + w_E= w_{T},\notag\\
% &&\hspace*{-2mm}\mathbf{C_6}: f_I + w_I \leq Sw + (S-1)w_G \leq f_E - w_E,\notag\\
&&\hspace*{-2mm} \mathbf{C_6}: w \leq \frac{f_I B_{\mathrm{th}} + w_I B_{\mathrm{th}} + (s-1)B_{\mathrm{th}}}{1-(0.5-s)B_{\mathrm{th}}}, \forall s \in \Sset,\notag\\
&&\hspace*{-2mm}\mathbf{C_7}:|k(f_E-w_E)-k(f_I+w_I)|\leq \epsilon,\notag\\
&&\hspace*{-2mm} \mathbf{C_8}:\assocNBS \in \{0,1\}, \forall n\in\Nset, b \in \Bset, s \in \Sset\notag.
\end{eqnarray}
In $\mathcal{P}_1$, $\mathbf{C_1}$ constrains the maximum transmit power limit at each BS, where $P_b^{\max}$ is the maximum transmit budget of $b$-th BS. {In order to leverage multi-connectivity, constraint $\mathbf{C_2}:\Bsum \assocNBS \leq 1, \forall s \in \Sset, n \in \Nset$ is taken into account, which limits a user to be associated with only one BS over the same sub-band, while allowing the association of a user with more than one BS across different sub-bands.}
% $\mathbf{C_2}$ ensures multi-connectivity, i.e., each user can be associated with multiple BSs on different sub-bands. 
$\mathbf{C_3}$ guarantees that each sub-band can only be assigned to one user in a given BS. 
$\mathbf{C_4}$ considers the fairness of users by imposing $\Gamma_{n}^L$ as the respective lower limit on the total number of sub-bands that can be assigned to a user. $\mathbf{C_5}$ constrains the total bandwidth of the TW of interest.
% and $\mathbf{C_6}$ ensures that the transmission sub-bands does not overlap or exceed the edge bands. 
$\mathbf{C_6}$ avoids the beam squint phenomena, and $\mathbf{C_7}$ ensures the molecular absorption coefficient do not vary beyond $\epsilon$ within the TW. The edge bands, i.e., $w_E$ and $w_I$, can guarantee this condition. $\mathbf{C_8}$ denotes that the joint user association and sub-band assignment variables are binary.

The problem $\pr_1$ is a mixed-integer non-linear programming (MINLP) problem and is thus non-convex. {The nonlinearity of the problem stems from the fact that the objective function is formulated as a fractional function of user association, sub-band assignment, and power allocation variables, along with their pairwise multiplication. Additionally, the problem is combinatorial due to the involvement of binary variables and non-convex due to the non-convexity of the molecular absorption coefficient in $\mathbf{C_7}$.} {An additional hurdle in tackling problem $\pr_1$ lies in determining the appropriate values for the number of sub-bands and their respective bandwidths. This complexity arises from the objective function, which involves a summation across the number of sub-bands. Consequently, finding the optimal values for $S$ and $w$ is not a straightforward task as the two variables are intertwined.} Furthermore, when considering a given $S$, the optimization of bandwidth cannot be conducted independently, as it is intrinsically tied to the total number of sub-bands. 
We decompose and solve the problem in two steps. First, we obtain a feasible solution for the bandwidths of the leading and trailing edge spectrum, the number of sub-bands, and their transmission bandwidths by proposing a convexification of the molecular absorption coefficient. Then, joint user association and sub-band assignment is performed in an optimal manner leveraging the uni-modularity of the problem followed by power optimization in an iterative manner using alternating optimization.

\subsection{{Molecular Absorption Coefficient}}
{
The molecular absorption coefficient, which is negligible in lower frequencies, is an important aspect of the THz band and can be different based on various environmental conditions, such as pressure, gas molecules, and so on \cite{MBN-Survey}. The molecular absorption coefficient of the isotopologue $t$ of gas $g$ for a molecular volumetric density at temperature $T$ and pressure of $p$ can be obtained as:   
\begin{equation}
    k(f) = \sum_{(t,g)} \frac{p^2 T_{\mathrm{sp}} q^{(t,g) } N_A S^{(t,g)} f \tanh{\left(\frac{h c f}{2 k_b T} \right)}}{p_0 T_k VT^2 f^{(t,g)}_c \tanh{\left(\frac{h c f^{(t,g)}_c}{2 k_b T} \right)}}F^{(t,g)}(f),
\end{equation}
where $p_0$ is the reference pressure at 1 atm, $T_{\mathrm{sp}}$ is the temperature at standard pressure, and the mixing ratio of gases is $q^{(t,g)}$. $N_A$ is Avogadro's constant, and $V$ is the gas constant. $S^{(t,g)}$ is the line intensity, showing the strength of the absorption by a specific type of molecule. $f$ denotes the frequency of the electromagnetic wave, and $f^{(t,g)}_{c}$ represents the resonant frequency of gas $g$. $h$ and $k_b$ denote the Planck's and Boltzmann's constants, respectively. For the frequency $f$, the Van Vleck-Weisskopf asymmetric line shape is given by:
\begin{equation}
    F^{(t,g)}(f) \!=\! \frac{100 c \alpha^{(t,g)}f}{\pi f^{(t,g)}_c} \!\! \left(\!\frac{1}{G^2 + \left(\alpha^{(t,g)} \right)^2} + \frac{1}{H^2 + \left(\alpha^{(t,g)} \right)^2} \! \right)\!,
\end{equation}
such that $G = f + f_c^{(t,g)}$, $H = f - f_c^{(t,g)}$, and the Lorentz half-width is obtained by:
\begin{equation}
    \alpha^{(t,g)} = \left(\left(1-q^{(t,g)} \right) \alpha^{(t,g)}_{\mathrm{air}} + q^{(t,g)} \alpha^{(t,g)}_0 \right) \left(\frac{p}{p_0} \right) \left(\frac{T_0}{T} \right)^l,
\end{equation}
where the air broadened half-widths, $\alpha^{(t,g)}_{\mathrm{air}}$, self-broadened half-widths, $\alpha^{(t,g)}_{0}$, and temperature broadening coefficient $l$ can be obtained directly from the high-resolution transmission molecular absorption (HITRAN) database \cite{HITRAN}.
}

\subsection{Convexifying the Molecular Absorption Coefficient}
{After obtaining the molecular absorption coefficient based on the previous section calculations and HITRAN database, we need to find an analytical expression for the molecular absorption coefficient as a function of frequency to enable us to derive the widths of the leading and trailing edge sub-bands, namely $w_I$ and $w_E$.} Additionally, using a convex function that spans the entire TW would be advantageous, allowing the use of numerical techniques like interior point methods to effectively address the feasibility problem in determining $w_I$ and $w_E$. Therefore, we propose the following convex approximation of the molecular absorption coefficient using curve-fitting techniques.
\begin{equation}\label{eqn:k(f)}
    \bar{k}(f) = t_1 \exp\left(\frac{-1}{{(t_2 f + t_3)}^2}\right) + t_4,
\end{equation}
where $t_1,t_4>0$, and $t_2,t_3$ are the fitted curve parameters. Table II presents values for curve-fitting parameters utilizing the suggested function for a variety of THz transmission windows. Note that $\bar{k}(f)$ is a convex function when $\frac{-\sqrt{6}-3t_3}{3t_2} \leq f \leq \frac{\sqrt{6}-3t_3}{3t_2}$, which can be satisfied in the TW of interest by selecting appropriate $f_I$ and $f_E$. Since the proposed function is twice differentiable as shown below, the convexity of the approximated molecular absorption coefficient can be verified.
\begin{equation}\label{eqn:ConvexityOfK(f)}
    \frac{d^2\bar{k}(f)}{d f^2} = \frac{-2t_1 t_2^2 \exp\left(\frac{-1}{{(t_2 f + t_3)}^2}\right)\left(3{(t_2 f + t_3)}^2-2\right)}{{(t_2 f + t_3)}^6}. \nonumber
\end{equation}
\subsection{Leading and Trailing Edge Bands of the TW}
{For a given TW, $f_I$, $f_E$, and curve-fitting parameters of the respective TW obtained from Table II, we can use the convexity of $\bar{k}(f)$ and solve the following set of inequalities 
derived from constraint $\mathbf{C_7}$ to obtain the feasible values of $w_E$ and $w_I$.
\begin{subequations}\label{eqn:MolecApprxSystemOfIneq}
\begin{align}
\bar{k}(f_E-w_E)-\bar{k}(f_I+w_I) \leq \epsilon \label{eqn:MolAbsAppr1}\\
\bar{k}(f_I+w_I)-\bar{k}(f_E-w_E) \leq \epsilon, \label{eqn:MolAbsAppr2}
\end{align}
\end{subequations}
{Considering $f_I + w_I < f_e - w_E$,} by using the interior-point method for solving \eqref{eqn:MolecApprxSystemOfIneq}, the feasible values of leading and trailing edge bands, i.e., $\hat{w}_I$ and $\hat{w}_E$ are obtained. After having $\hat{w}_I$ and $\hat{w}_E$ obtained, we determine the total number of sub-bands and their respective bandwidth as discussed in the next subsection.}
% \begin{table}[]
% \caption{Curve-fitting parameters based on the approximated function of molecular absorption coefficient, $\bar{k}(f)$.}
% \centering
% \begin{tabular}{|c|c|c|c|c|}
% \hline
% \textbf{TW Range (THz)} & $t_1$ & $t_2$    & $t_3$   & $t_4$  \\ \hline
% 0.448 - 0.541                            & 1.1   & -14.5233 & 7.1063  & 0.0173 \\ \hline
% 0.624 - 0.742                            & 0.8   & 11.3600  & -7.6442 & 0.0139 \\ \hline
% 0.76 - 0.915                             & 0.5   & 9.6221   & -8.1526 & 0.0139 \\ \hline
% 0.997 - 1.07                             & 1.2   & -21.6372 & 22.28   & 0.0882 \\ \hline
% \end{tabular}
% \label{tab:Curve-fitting}\vspace{-5mm}
% \end{table}
\begin{table}[]
\caption{Curve-fitting parameters based on the approximated function of molecular absorption coefficient, $\bar{k}(f)$.}
\centering
\begin{tabular}{|c|c|c|c|c|c|}
\hline
{\textbf{TW Label}} & \textbf{TW Range (THz)} & $t_1$ & $t_2$    & $t_3$   & $t_4$  \\ \hline
{$\text{TW}_1$} & 0.448 - 0.531                            & 1.1   & -14.5233 & 7.1063  & 0.0173 \\ \hline
{$\text{TW}_2$} & 0.624 - 0.722                            & 0.8   & 11.3600  & -7.6442 & 0.0139 \\ \hline
{$\text{TW}_3$} & 0.78 - 0.915                             & 0.5   & 9.6221   & -8.1526 & 0.0139 \\ \hline
{$\text{TW}_4$} & 0.997 - 1.063                             & 1.2   & -21.6372 & 22.28   & 0.0882 \\ \hline
\end{tabular}
\label{tab:Curve-fitting}
\end{table}

\subsection{Bound on the Number of THz  Sub-bands with Beam-Squint}
{Acquiring the optimal $S$ and $w$ presents a challenge as they are interwined and due to the inclusion of a summation over the variable $S$ in the objective function of problem $\pr_1$.}
%\textcolor{blue}{
Hence, for the given $\hat{w}_I$ and $\hat{w}_E$, a feasible lower bound on the number of sub-bands required to overcome beam-squint can be obtained using constraint $\mathbf{C_6}$. We first recast constraint $\mathbf{C_6}$ as
$\frac{\frac{w}{B_{\mathrm{th}}}-f_I - \hat{w}_I + 0.5 w + w_G}{w+w_G} \leq s, \forall s \in \Sset.
$
When $\mathbf{C_6}$ is satisfied for $\min(\Sset) = 1$, it avoids the beam-squint at all sub-bands, thus by substituting $s=1$, we get:
\begin{equation}\label{eqn:BeamSplit3}
    \frac{w}{B_{\mathrm{th}}} - f_I - \hat{w}_I -0.5 w \leq 0.
\end{equation}
By substituting $w = \frac{w_T-\hat{w}_E-\hat{w}_I-(S-1)w_G}{S}$ from $\mathbf{C_5}$ into \eqref{eqn:BeamSplit3}, the lower bound on $S$ is given by:
\begin{equation}\label{eqn:LowerBound-S}
    S_{\mathrm{LB}} \!\geq \!\!\Big\lceil\frac{2(w_T\!-\!\hat{w}_I\!+\!w_G\!-\!\hat{w}_E)\!+\!B_{\mathrm{th}}(-w_T\!+\!\hat{w}_I\!-\!w_G\!+\!\hat{w}_E)}{2w_G\!+\!2B_{\mathrm{th}}(f_I+\hat{w}_I)+B_{\mathrm{th}}w_G}\Big\rceil.
\end{equation}
{
Then, using $S_{\mathrm{LB}}$, the respective bandwidth is given as follows:
\begin{equation}\label{eqn:Bandwidth}
    w = \frac{w_T-\hat{w}_E-\hat{w}_I-(S_{\mathrm{LB}}-1)w_G}{S_{\mathrm{LB}}}.
\end{equation}
\begin{lemma}\label{lemma1_Decreasing_S}
    For a given power allocation and user association, the objective function in $\pr_1$, i.e., $\Bsum \Nsum \Ssum R_{b,s,n}$ is a decreasing function of the total number of sub-bands, $S$. Therefore, the lower bound on $S_{\mathrm{LB}} = S^*$ is optimal.
\end{lemma}
\begin{proof}
    Let us assume that $\lceil \alpha S \rceil \leq S$ sub-bands are assigned to a given user, such that $\alpha \in [0,1]$. Then, the achievable rate of the user is given as follows:
    \begin{align}
            &R_{n} \!=\! \Bsum \!\sum_{s=1}^{\lceil \alpha S \rceil}\! \frac{\bar{w}\!-\!(S\!-\!1)w_G}{S}\!\log\! \bigg(\!1\!+\! \frac{\frac{B P_b^{\max}}{\lceil \alpha S \rceil}}{\frac{B P_b^{\max}}{\lceil \beta S \rceil}\! + \!\!N_0(\frac{\bar{w}\!-\!(S\!-\!1)w_G}{S})}\!\bigg) \notag \\
            & = \!B \lceil \alpha S \rceil \frac{\bar{w}\!-\!(S\!-\!1)w_G}{S}\!\log\! \bigg(\!1\!+\!\frac{\frac{B P_b^{\max}}{\alpha}}{\frac{B P_b^{\max}}{\beta} \!+\!\! N_0(\bar{w}\!-\!(S\!-\!1)w_G)}\!\bigg)\notag \\
             &= f_{W}(S) f_R(S),
    \end{align}
    where $\bar{w} = w_T-\hat{w}_E-\hat{w}_I$. $\alpha$ and $\beta$ determine the portion of BSs and sub-bands assigned to a given user. It can be seen that $f_W(S)$ is a decreasing function of $S$, while $f_R(S)$ is an increasing function of $S$. However, due to the near-to-zero value of $N_0$, $N_0(\bar{w}-(S-1)w_G)$ is negligible compared to the first term in the denominator, even for large values of $S$. Also, the absolute value of $f_W(S)$ is always greater than that of $f_R(S)$, i.e., $\lvert{f_W(S)}\rvert >> \lvert{f_R(S)}\rvert$ due to the large bandwidths in the THz band. Therefore, the product of these two functions is monotonically decreasing with respect to $S$. Based on the similar discussion for the rest of the users, we can conclude that the objective function $\Bsum \Nsum \Ssum R_{b,s,n}$ is a decreasing function of $S$.
\end{proof}
Based on the result of \textbf{Lemma \ref{lemma1_Decreasing_S}}, we can select $S_{\mathrm{LB}}$ as the optimal value for the total number of sub-bands and their respective bandwidth using \eqref{eqn:Bandwidth}.
}
%%%%%% Uppder bound on S
% Similarly, constraint $\mathbf{C_6}$ can be reformulated as
% $
%     f_I \leq w_T -2 w_I \leq f_E.
% $
% On the other hand, the bandwidth of each transmission sub-band should satisfy $w \geq \beta w_G$, such that $\beta \geq 1$ adjusts the ratio of bandwidth and guard bands. Therefore, the upper bound on the total number of sub-bands can be formulated as follows:
% \begin{equation}\label{eqn:S-upper}
%     S_{\mathrm{UB}} \leq \frac{w_T-w_E-w_I+w_G}{(1+\beta)w_G}.
% \end{equation}
%%%%%% Uppder bound on S

% By using the expression in \eqref{eqn:LowerBound-S}, which alleviate the beam squint phenomena and is a function of $w_I$ and $w_E$, the following lower bound on the total number of sub-bands is stated.
% \begin{equation}\label{eqn:S-StartingPoint}
%     \hat{S} \!=\!\! {\left\lceil\frac{2(w_T\!-\!\hat{w}_I\!+\!w_G-\hat{w}_E)\!+\!B_{\mathrm{th}}(-w_T\!+\!\hat{w}_I\!-\!w_G\!+\!\hat{w}_E)}{2w_G\!+\!2B_{\mathrm{th}}(f_I+\hat{w}_I)+B_{\mathrm{th}}w_G}\right\rceil},
% \end{equation}
% where $\hat{w}_I$ and $\hat{w}_e$ are the edge band values derived using $\bar{k}(f)$, and $\lceil x \rceil$ denotes the ceiling function.

\subsection{Transformation of the Objective Function}
The optimization problem $\pr_1$ is MINLP and the non-linearity comes from the objective function since $\gam(\A,\p)$ is a non-convex function of the joint power allocation, user association, and sub-band assignment variables. To handle the non-convexity of the objective function, we transform the objective function using the following Lemma.
\begin{lemma}\label{Lemma2_3D_To_2D}
By reducing the dimension of the power allocation variables from $B\times S\times N$ to $B\times S$ using $\Nsum \assocNBS \pNBS = {p}_{b,s}$, the objective function in $\pr_1$ can  be written as $f(\A,\p) = \Bsum\Nsum\Ssum \assocNBS w \log(1+\gam(\p))$.
\end{lemma}
\begin{proof}
Considering the fact that each BS transmits to only one user over a given sub-band, i.e., equation in $\mathbf{C_3}$, we have $\Nsum \assocNBS \pNBS = {p}_{b,s}$. Also, 
$
    \Npsum a_{\bp,s,\np}p_{\bp,s,\np} = {p}_{\bp,s} - a_{\bp,s,n}p_{\bp,s,n}.
$
The cumulative interference can thus be written as follows:
\begin{equation}\label{eqn:2D-Interference}
\bar{I}_{b,s,n}(\p,\A) \!=\!\!\! \Bpsum\!\! q{p}_{\bp,s} \abss{\bar{h}_{\bp,s,n}}\!-\!\!\Bpsum\!\! q a_{\bp,s,n} p_{\bp,s,n} \abss{\bar{h}_{\bp,s,n}}.
\end{equation}
Since each user can only associate with one BS over a given sub-band, i.e., equation in $\mathbf{C_2}$,
the second term in \eqref{eqn:2D-Interference} is zero if the association in the numerator of $\gam$ is one, due to having $\assocNBS + \Bpsum a_{\bp,s,n} = 1$. Otherwise, the SINR and the rate are zero. A similar analogy can be used for the numerator, hence we can take the binary variable out of the logarithm and both the numerator and denominator while considering the 2D power allocation. Thus, the objective function can be replaced with $f(\A,\p) = \Bsum\Nsum\Ssum \assocNBS w \log(1+\gam(\p))$.
\end{proof}

\section{Proposed Algorithm for WSR Maximization in a Multi-Cell THz Network}
In this section, we present a two-stage algorithm, where the initial stage determines the edge bands, total number of sub-bands, and bandwidth as described in the previous section. The second stage  iteratively optimizes the joint user association and sub-band assignment, and power allocation variables. 
\subsection{Joint User Association and Sub-band Assignment} 
Given $\{{S^*},w,\hat{w}_I,\hat{w}_E,\p\}$, $\pr_1$ can be rewritten as follows:
% At this point, we solve the optimization problem $\pr_1$ for given set of $\{\boldsymbol{p}, w, S, \bar{w}_I, \bar{w}_e\}$ to obtain joint user association and sub-band assignment. First, in order to make the problem tractable, we leverage the binary variables and take $\assocNBS$ out of the logarithm function. That is, when the power allocations are known, and $\assocNBS = 1$, achievable rate, $R_{n,b,s}$ is a logarithmic function and if $\assocNBS = 0$, we have $R_{n,b,s} = 0$ \cite{Binary-Reformulation}.
% Therefore, to acquire joint user association and sub-band assignment, it is required to solve the following integer programming problem.
% \begin{maxi!}[3]
% 	{\A}{\!\!\!\!\!\Bsum\Nsum\Ssum \omega_{n,b} w \log(1+\gam(\A,\p)) \label{Con:p2Obj}}
% 	{\label{problem:p2main}}{\pr_{2}: \ \ }
%     \addConstraint{\Bsum \assocNBS \leq 1, \forall s \in \Sset, n \in \Nset}{}\label{Con:p2Assoc1}
%     \addConstraint{\Nsum \assocNBS = 1, \forall b \in \Bset, s \in \Sset}{}\label{Con:p2Assoc2}
%     \addConstraint{\Gamma_{n}^L \leq \Bsum \Ssum \assocNBS \leq \Gamma_{n}^U}{, \forall n \in \Nset}\label{Con:p2Fairness}
%     \addConstraint{\assocNBS \in \{0,1\},}{\forall n\in\Nset, b \in \Bset, s \in \Sset}\label{Con:p2Integer}
% \end{maxi!}
\begin{eqnarray}
\label{prob2} \pr_2:\hspace*{-4mm}
&&\hspace*{2mm}\underset{\A}{\textrm{maximize}} \,\, \,\, \Bsum\Nsum\Ssum w \log(1+\gam(\A))\\
\mbox{s.t.}\hspace*{-4mm}
&&\mathbf{C_2},\ \mathbf{C_3},\ \mathbf{C_4},\  \mathbf{C_8}.\notag
\end{eqnarray}
% The optimization problem $\pr_2$ is a MINLP. The non-linearity comes from the objective function since $\gam(\A,\p)$ is a non-convex function of the joint association and assignment variables. First, in order to handle the non-convexity of the objective function, we transform the power allocation variables, i.e., $\p$ from 3-dimensional (3D) variables to 2D variables using the following Lemma.
{
Similar to the proof of Lemma~\ref{Lemma2_3D_To_2D}, employing constraint $\mathbf{C_3}$ results in $\Nsum \assocNBS \pNBS = {p}_{b,s}$, causing the power budget constraint in $\mathbf{C_1}$ to become $\Ssum p_{b,s}\leq P_b^{\max}, \ \forall b$. Consequently, $\mathbf{C_1}$ can be eliminated from the problem of joint user association and sub-band assignment since it no longer contains the binary variable.} 
Due to the unimodular structure of $\pr_2$, we can solve the problem \textit{optimally} without utilizing greedy or matching solutions. To this end, we need to define $\mathbf{a} = [{a}_1,\dots,{a}_i,\dots,{a}_{BSN}]^T \in {\{0,1\}}^{BSN\times 1}$, consisting of all association and assignment variables, such that ${a}_i = \assocNBS$, and $i=n+N(s-1)+NS(b-1)$. In the following, all constraints are written in matrix format. 
\begin{subequations}\label{eqn:Constraints_Matrix}
\begin{align}
&\mathbf{C}:[\mathbf{I}_{NS\times NS},\dots,\mathbf{I}_{NS\times NS}] = \mathbf{1}^{T}_{B\times 1}\otimes \mathbf{I}_{NS\times NS}, \\
&\mathbf{D}: \text{diag}(\mathbf{L},B), \ \mathbf{L} = \mathbf{I}_{S\times S}\otimes \mathbf{1}^T_{N\times 1},\\
&\mathbf{E}:[\mathbf{I}_{N\times N},\dots,\mathbf{I}_{N\times N}]=\mathbf{1}^T_{BS\times 1} \otimes \mathbf{I}_{N\times N}, \\
&\mathbf{F}: [-\mathbf{I}_{N\times N},\dots,-\mathbf{I}_{N\times N}]=-\mathbf{1}^T_{BS\times 1} \otimes \mathbf{I}_{N\times N}=-\mathbf{E}.
\end{align}
\end{subequations}
In \eqref{eqn:Constraints_Matrix}, $\mathbf{C}$ and $\mathbf{D}$ represent constraints $\mathbf{C_2}$ and $\mathbf{C_3}$, respectively. Matrices $\mathbf{E}$ and $\mathbf{F}$ can be obtained by dividing the inequalities in constraint $\mathbf{C_4}$ into two constraints, such that $\mathbf{E}:  \Bsum \Ssum \assocNBS \leq S, \forall n \in \Nset$, and $\mathbf{F}: -\Bsum \Ssum \assocNBS \leq -\Gamma_{n}^L, \forall n \in \Nset$. %Then, the matrices of all constraints are concatenated as $\mathbf{T} = [\mathbf{C},\mathbf{D},\mathbf{E},\mathbf{F}]^T$.
Moreover, the right-hand side of the above constraints can be reformulated as $\mathbf{k}_1 = \mathbf{1}_{NS\times 1}$, $\mathbf{k}_2 = \mathbf{1}_{BS\times 1}$, $\mathbf{k}_3 = [S,\dots,S] \in \mathbb{N}^{NB\times 1}$, and $\mathbf{k}_4 = [-\Gamma^L_{1},\dots,-\Gamma^L_{N}]=\mathbf{\Gamma}^L \in \mathbb{N}^{NB\times 1}$.
Finally, using the above-mentioned, the optimization problem $\pr_2$ can be reformulated as follows:
% \begin{maxi!}[3]
% 	{\mathbf{a}}{\!\!\!\!\!\sum\limits_{i=1}^{BSN} \mathbf{q}^T\mathbf{a} \label{Con:p3Obj}}
% 	{\label{problem:p3main}}{\pr_{3}: \ \ }
%     \addConstraint{\mathbf{C}\mathbf{a}\preceq \mathbf{k}_1}
%     {}\label{Con:p3C}
%     \addConstraint{\mathbf{D}\mathbf{a}= \mathbf{k}_2}{}\label{Con:p3D}
%     \addConstraint{\mathbf{E}\mathbf{a}\preceq \mathbf{k}_3}{}\label{Con:p3E}
%     \addConstraint{\mathbf{F}\mathbf{a}\preceq \mathbf{k_4}}{}\label{Con:p3F}
%     \addConstraint{{a}_i \in \{0,1\},}{\forall i \in \{1,\dots,BSN\},}\label{Con:p3Integer}
% \end{maxi!}
\begin{eqnarray}
\label{prob3} \pr_3:\hspace*{-4mm}
&&\hspace*{2mm}\underset{\mathbf{a}}{\textrm{maximize}} \,\, \,\, \sum\limits_{i=1}^{BSN} \mathbf{q}^T\mathbf{a}\\
\mbox{s.t.}\hspace*{-4mm}
&& \hspace{-1mm}{\mathbf{C}}_\mathbf{2}:\mathbf{C}\mathbf{a}\preceq \mathbf{k}_1, \hspace*{1mm}{\mathbf{C}}_{\mathbf{3}}:\mathbf{D}\mathbf{a}= \mathbf{k}_2, \hspace*{1mm}{\mathbf{C}}_{\mathbf{4a}}:\mathbf{E}\mathbf{a}\preceq \mathbf{k}_3, \notag \\
&&\hspace{-1mm}{\mathbf{C}}_{\mathbf{4b}}: \mathbf{F}\mathbf{a}\preceq \mathbf{k_4}, \hspace*{1mm}
{\mathbf{C}}_\mathbf{8a}:{a}_i \in \{0,1\}, \forall i \in \{1,\dots,BSN\}
\notag,
\end{eqnarray}
where $\mathbf{q}=[ w \log(1+{\gamma}_i(\boldsymbol{p}))]^T, \forall i \in \{1,\dots,BSN\}$, such that ${\gamma}_i$ is the $i$-th element of the vector form of its respective matrix and $i$ is obtained similar to the index of $a_i$. In what follows, leveraging the concepts of unimodularity, we prove that the integer linear programming problem $\pr_3$ can be converted into a  linear programming problem without loss of optimality, and the solutions are integer values. To this end, we first define $\mathbf{T}=[\mathbf{C}^T,\mathbf{D}^T,\mathbf{E}^T,\mathbf{F}^T]^T$ as the constraint matrix of $\pr_3$, and $\mathbf{k}=[\mathbf{k}_1^T,\mathbf{k}_2^T,\mathbf{k}_3^T,\mathbf{k}_4^T]^T$.

\begin{definition} [Totally Unimodular Matrix (TUM)]
    A matrix $\mathbf{M}$ is TUM if each square submatrix has a determinant of 0, 1, or -1, i.e., all entries are either 0 or $\pm$ 1 \cite{schrijver1998theory}. 
\end{definition}

\begin{theorem}\label{Theo_1}
     If the constraint matrix $\mathbf{T}$ is totally unimodular, and $\mathbf{k}$ is integer, problem $\pr_3$ can be relaxed by removing the integer constraint, and the resulting problem, which is linear programming, can be solved via the dual-simplex method without loss of optimality.  
\end{theorem}
\begin{proof}
    See Theorem (Unimodularity implies integrality) in \cite{lee2004first}.
\end{proof}
In the following, we first perform some operations on the matrix ($\mathbf{T}$) that preserve TUM (as described below) leading to an equivalent matrix $\mathbf{T}^{(3)}$. We then show that  $\mathbf{T}^{(3)}$  is TUM  which in turn implies that  $\mathbf{T}$ is TUM. 
% Then, a necessary and sufficient condition for total unimodularity is provided. Finally, the matrix describing the constraints of problem $\pr_3$ ($\mathbf{T}$)  is simplified, visualized, and then proved to be totally unimodular.
 \begin{definition}[Operations that Preserve Unimodularity]
     Total unimodularity is preserved under transposition,  permutation of rows or columns, adding or removing repetitive rows or columns,  scaling rows or columns by $-1$, and pivoting over any non-zero element \cite{schrijver1998theory}.
 \end{definition}
Consider the constraint matrix $\mathbf{T}=[\mathbf{C}^T,\mathbf{D}^T,\mathbf{E}^T,\mathbf{F}^T]^T$. From the definition of matrix $\mathbf{E}$ and $\mathbf{F}$, we know that $\mathbf{F}= -\mathbf{E}$ (see (\ref{eqn:Constraints_Matrix}d)). Thus, $\mathbf{T}$ can be written as $\mathbf{T} = [\mathbf{C}^T, \mathbf{D}^T, \mathbf{E}^T, -\mathbf{E}^T]^T$. Using \textbf{Definition~2}, we can multiply all rows of matrix $\mathbf{F}$ by -1, resulting in  $\mathbf{T}^{(1)} = [\mathbf{C}^T, \mathbf{D}^T, \mathbf{E}^T, \mathbf{E}^T]^T$. Also, since matrix $\mathbf{E}$ appeared twice, we can eliminate it, which results in the  matrix $\mathbf{T}^{(2)} = [\mathbf{C}^T, \mathbf{D}^T, \mathbf{E}^T]^T$ (as shown in the figure). 
Next, we apply the pivoting operation (defined below) to the non-zero elements of the first $\mathbf{I}$ matrix of $\mathbf{C}$. 
\begin{definition}{Pivoting Operation \cite{schrijver1998theory}}: 
Consider the following $m \times n$ matrix $\mathbf{A}$, the pivoting operation over a non-zero element, like $e$, results in a matrix $\mathbf{B}$ that has the form:
\begin{equation}
A = \left[\begin{array}{c c}
e & \mathbf{c}^T \\ 
\mathbf{b} & \mathbf{D}
\end{array}\right] \longrightarrow
B = \left[\begin{array}{c c}
-e & e\mathbf{c}^T \\ 
e\mathbf{b} & \mathbf{D} - e\mathbf{b}\mathbf{c}^T
\end{array}\right]
\end{equation}
where $\mathbf{b}$ and $\mathbf{c}$ are $(m-1)$ and $(n-1) $ dimensional column vectors, and $\mathbf{D}$ is a $(m-1)\times(n-1)$ matrix. Note that pivoting over the element $-e$ of matrix $\mathbf{B}$ results in the original matrix $\mathbf{A}$. The pivoting operation is critical in establishing the total unimodularity of matrix $T$. 
\end{definition}

\begin{figure}[t!]
\begin{minipage}{0.42\textwidth}
    \includegraphics[scale=0.38]{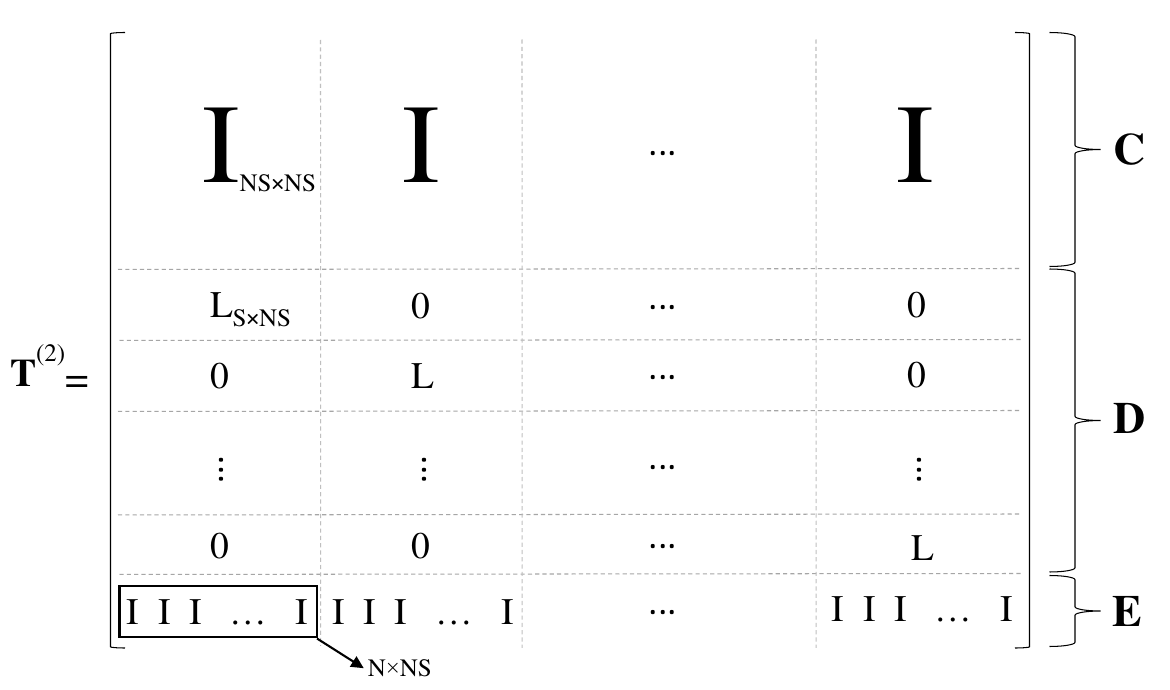}
    \label{T_2}
\end{minipage}
\hfill
\begin{minipage}{0.49\textwidth}
\includegraphics[scale=0.38]{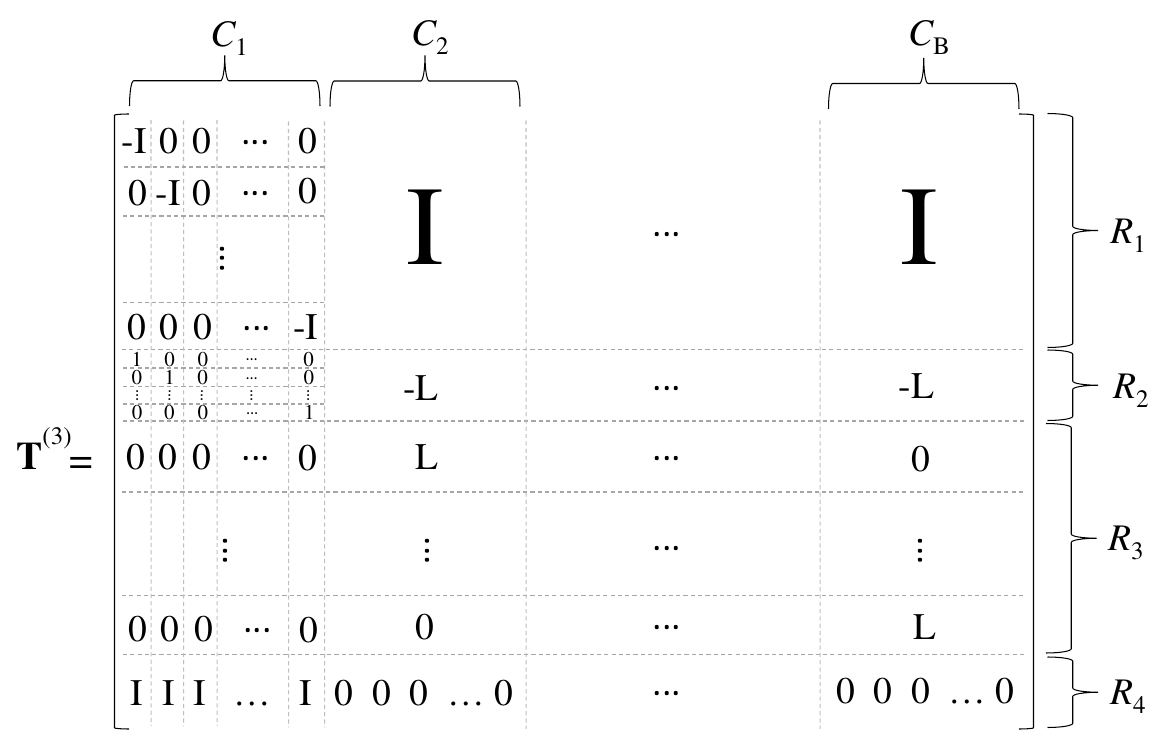}
\label{T_3}
\end{minipage}\vspace{-2mm}
\caption{Graphical illustration of $\mathbf{T}^{(2)}$ and $\mathbf{T}^{(3)}$.} \label{Unimodular_Mat}
\vspace{-4mm}
\end{figure}

After pivoting, the resulting matrix, $\mathbf{T}^{(3)}$, has the structure (as shown in Fig.~2). Since the transformation from $\mathbf{T}$ to $\mathbf{T}^{(3)}$ uses the operations mentioned in \textbf{Definition~2}, and they preserve total unimodularity, it suffices to show that $\mathbf{T}^{(3)}$ is TUM, which in turn proves that $\mathbf{T}$ is TUM. 
\begin{lemma}\label{Lemma_2}
    Matrix $\mathbf{T_3}$ and in turn $\mathbf{T}$ is TUM and $\mathbf{k}$ is integer. Thus, $a_i \in \{0,1\}$ becomes $a_i \geq 0$ and the resulting linear program has the same integer optimum point as $\pr_3$.
\end{lemma}
\begin{proof}
    See Appendix A.
\end{proof}
According to \textbf{Theorem 1} and \textbf{Lemma 3}, the integer programming in $\pr_3$ can be efficiently solved in polynomial time using the dual-simplex method, and the joint user association and sub-band assignment matrix $\A$ can be reconstructed using vector $\mathbf{a}$.
{Table-\ref{tab:Time_Complexity_Assoc} compares the time complexity of the proposed linear programming (LP) and binary solution (Branch-and-bound method using Mosek solver \cite{mosek}).}
\begin{table*}[]
\centering
\caption{Comparison of time complexity between the proposed LP and binary solution of $\mathcal{P}_3$}
%\resizebox{10}{!}{
\begin{tabular}{|c|l|l|l|l|}
\hline
\textbf{Problem Size}    & $B=2,S=2,N=2$    & $B=4,S=8,N=6$    & $B=6,S=10,N=12$  & $B=15,S=20,N=30$ \\ \hline
\textbf{Proposed LP}     & 0.0090 {[}Sec{]} & 0.0109 {[}Sec{]} & 0.0129 {[}Sec{]} & 0.1268 {[}Sec{]} \\ \hline
\textbf{Binary Solution} & 0.1794 {[}Sec{]} & 0.2882 {[}Sec{]} & 0.4124 {[}Sec{]} & 2.4470 {[}Sec{]} \\ \hline
\end{tabular}
\label{tab:Time_Complexity_Assoc}
%\vspace{-5mm}
\end{table*}

\subsection{Centralized Power Allocation with Fractional Programming and Bisection Search}
Given $\{{S^*},w,\hat{w}_I,\hat{w}_E,\A\}$, $\pr_1$ must be solved for $\p$, which is a non-convex problem due to the variables in both numerator and denominator of the objective function. Solving a such problem is NP-hard and finding a globally optimum solution is computationally prohibitive. Applying the change of variables $\bar{p}_{b,s}=\sqrt{p_{b,s}}, \forall n \in \Nset$, the optimization problem can be formulated as:
% \begin{maxi!}[3]
% 	{\bar{\p},\gamma_{b,s,n}}{f(\A,\bar{\boldsymbol{P}},\gam) = \Bsum \Nsum \Ssum \assocNBS\omega_{n,b} w \log(1+\gamma_{b,s,n}) \label{Con:p4Obj}}
% 	{\label{problem:p4main}}{\mathcal{P}_4: \ \ }
% 	\addConstraint{\Ssum \bar{p}_{b,s}^2}{\leq P_b^{\max}, \forall b \in \Bset}\label{Con:p4Power}
% 	\addConstraint{\gamma_{b,s,n} \! = \! \frac{\bar{p}^2_{b,s}\h }{I_{b,s,n} + \phi_{b,s,n} + \Phi_{b,s,b} + \bar{N_0} f_s^2 w}}{}\label{Con:p4SINR}
% \end{maxi!}
\begin{eqnarray}
\label{prob4} \pr_4:\hspace*{-4mm}
&&\hspace*{-2mm}\underset{\bar{\p},\gamma_{b,s,n}}{\textrm{maximize}} \,\, \,\,\hspace*{-2mm} f(\bar{\p},\gam) \!=\!\! \Bsum \Nsum \Ssum \! \assocNBS w \log(1\!+\!\gamma_{b,s,n}) \notag\\
\mbox{s.t.}\hspace*{-4mm}
&&{\mathbf{C}}_\mathbf{1}:\Ssum \bar{p}_{b,s}^2\leq P_b^{\max}, \forall b \in \Bset, \\&&{\mathbf{C}}_{\mathbf{9}}:\gamma_{b,s,n} \! = \! \frac{\bar{p}_{b,s}^2 \h}{\Bpsum q \bar{p}_{\bp,s}^2\abss{\bar{h}_{\bp,s,n}} \! + \! \bar{p}_{b,s}^2 {\lvert{\tilde{h}_{b,s,n}}\rvert}^2 + N_0 w} \notag.
\end{eqnarray}
Note that the summations over the users and the binary variables in ${\mathbf{C}}_\mathbf{1}$ are omitted according to \textbf{Lemma~2}. 
{We first apply the Lagrangian dual transform to relocate the SINR expression within $\gam$ to the outside of the logarithm. By doing so, we are able to apply the quadratic transformation and solve the problem iteratively with closed-form expressions. Subsequently, we employ the quadratic transformation to convert the SINR into a concave quadratic function.}

In order to apply the Lagrangian dual transform, for fixed $\bar{\p}$, the Lagrangian function of $\pr_4$ with respect to  ${\mathbf{C}}_{\mathbf{9}}$ can be written as follows:
\begin{equation*}
    \hspace{-1.5cm}\mathcal{L}(\bar{\p},\boldsymbol{\gamma},\boldsymbol{\lambda}) = \!\Bsum \Nsum\Ssum \Big[\assocNBS w\log(1\!+\!\gam)\!
\end{equation*}
\begin{equation}\label{eqn:Lag1}
    -\!\lambda_{b,s,n}\!\!\left(\!\!\gam \!-\!\frac{\bar{p}^2_{b,s}\h }{\Bpsum\!\! q \bar{p}^2_{\bp,s}\hbpbar \!\!+\! \bar{p}^2_{b,s}\htilde\!\!\! +\!\! {N_0} w} \!\!\right)\Big], 
\end{equation}
where $\boldsymbol{\gamma} \in \mathbb{R}_+^{B S\times N} $ is the matrix of all $\gam$ and $\boldsymbol{\lambda} \in \mathbb{R}_+^{B S\times N}$ is the matrix of all Lagrange multipliers $\lambda_{b,s,n}$. Given $\bar{\p}$, the first-order condition is satisfied with respect to $\gam$ when $\frac{\partial \mathcal{L}(\bar{\p},\boldsymbol{\gamma},\boldsymbol{\lambda})}{\partial \gam} =0 $, which leads to
\begin{equation}\label{eqn:LagMultp1}
%\begin{split}
    \lambda_{b,s,n}^{*} \!=\!  \frac{\assocNBS w\!\left(\Bpsum q \bar{p}^2_{\bp,s}\hbpbar\!\! + \!\bar{p}^2_{b,s}\h \!\!+\!\! {N_0} w\right)}{\Bpsum\! q \bar{p}^2_{\bp,s}\hbpbar\!\! +\! \bar{p}^2_{b,s}(\h\!\!+\!\!\htilde) \!+\! {N_0} w}
%\end{split}
\end{equation}
and
\begin{equation}\label{eqn:OptGamma}
    \gam^{*} = \frac{\bar{p}^2_{b,s}\h }{\Bpsum q \bar{p}^2_{\bp,s}\hbpbar + \bar{p}^2_{b,s}\htilde + {N_0} w}
\end{equation}
Next, by substituting $\lambda_{b,s,n}^{*}$ into the Lagrange function in \eqref{eqn:Lag1}, we have:
\begin{equation}\label{eqn:Obj-Lagrange1}
    \begin{split}
        & {f}_1(\bar{\p},\boldsymbol{\gamma})\! = \!\!\Bsum \Nsum \Ssum\! \big[ \assocNBS w \log(1\!+\!\gam) \!-\!  \assocNBS w \gam \\
        & + \frac{\assocNBS w(1+\gam)\bar{p}^2_{b,s}\h}{\Bpsum q \bar{p}^2_{\bp,s}\hbpbar + \bar{p}^2_{b,s}(\h+\htilde) + {N_0} w}\big].
    \end{split}
\end{equation}
According to \cite{FP-Part1,FP-Part2},  strong duality holds for the objective function in \eqref{prob4}. Thus, the transformed objective function ${f}_1(\bar{\p},\boldsymbol{\gamma})$ can be maximized instead of $f(\bar{\p},\boldsymbol{\gamma})$. The transformed objective function
${f}_1(\bar{\p},\boldsymbol{\gamma})$ is still a non-convex function of $\bar{\p}$ due to the sum-of-ratios. Hence, the following quadratic transformation is applied to handle the non-convexity. 
\begin{equation*}
    {f}_2(\bar{\p},\boldsymbol{\gamma},\boldsymbol{Y}) \!=\!\!\! \Bsum \Nsum \Ssum\! \big[\assocNBS w [\log(1\!+\!\gam)\!-\!\gam]
\end{equation*}
\begin{equation*}
    \hspace{-2.5cm}+2y_{b,s,n}\sqrt{\assocNBS w(1\!+\!\gam)\bar{p}_{b,s}^2 \h}
\end{equation*}
\begin{equation}\label{eqn:ObjFunc2}
    - y_{b,s,n}^2 \big(\!\!\Bpsum\!\! q \bar{p}_{\bp,s}^2 \abss{\bar{h}_{\bp,s,n}} \! + \bar{p}_{b,s}^2 (\h + \htilde) + {N}_0 w\big)\big],
\end{equation}
where $\mathbf{Y} \in \mathbb{R}_+^{B S\times N}$ is the matrix of  auxiliary variables $y_{b,s,n}$, and ${f}_2(\bar{\p},\boldsymbol{\gamma},\boldsymbol{Y})$ is a concave function of $\boldsymbol{Y}$.
\begin{remark}
The optimal solution of $\underset{\boldsymbol{Y}}{\max} {f}_2(\bar{\p},\boldsymbol{\gamma},\boldsymbol{Y})$, i.e., ${\boldsymbol{Y}}^*$, results in ${f}_1(\bar{\p},\boldsymbol{\gamma})$. The reason is that by setting the partial derivative of ${f}_2(\bar{\p},\boldsymbol{\gamma},\boldsymbol{Y})$ with respect to $\boldsymbol{Y}$, and substituting the resulting $\boldsymbol{Y}$ into \eqref{eqn:ObjFunc2}, the objective function ${f}_1(\bar{\p},\boldsymbol{\gamma})$ can be obtained.
\end{remark}
\noindent Based on  \textit{Remark 1}, for fixed $\bar{\p}$ and $\boldsymbol{\gamma}$,  $\boldsymbol{Y}^*$ can be found by setting $\frac{\partial{f}_2(\bar{\p},\boldsymbol{\gamma},\boldsymbol{Y})}{\partial y_{b,s,n}}=0$, i.e.,
\begin{equation}\label{eqn:Optimal-Y}
        y_{b,s,n}^{*} =  \frac{\sqrt{\assocNBS w(1+\gam)\bar{p}_{b,s}^2 \h}}{\Bpsum\!\! q \bar{p}_{\bp,s}^2 \abss{\bar{h}_{\bp,s,n}}\!\! +\bar{p}_{b,s}^2 (\h\!\! +\! \htilde)\! +\! {N}_0  w}.
\end{equation}
Then, to obtain optimal power allocation given $\boldsymbol{\gamma},\boldsymbol{Y}$, we rewrite $\pr_4$ as:
% \begin{maxi!}[3]
% 	{\bar{\p}}{\tilde{f}_2(\A,\bar{\p},\boldsymbol{\gamma},\boldsymbol{Y}) \label{Con:p5Obj}}
% 	{\label{problem:p5main}}{\pr_{5}: \ \ }
% 	\addConstraint{\Ssum  \bar{p}_{b,s}^2}{\leq P_b^{\max}, \forall b \in \Bset}\label{Con:p5Power}
% \end{maxi!}
\begin{eqnarray}
\label{prob5} \pr_5:\hspace*{-4mm}
&&\hspace*{2mm}\underset{\bar{\p}}{\textrm{maximize}} \,\, \,\, {f}_2(\bar{\p},\boldsymbol{\gamma},\boldsymbol{Y})\\
\mbox{s.t.}\hspace*{-4mm}
&&\mathbf{C_1}:\Ssum \bar{p}_{b,s}^2\leq P_b^{\max}, \forall b \in \Bset\notag.
\end{eqnarray}
The optimization problem $\pr_{5}$ is a convex problem and a closed-form expression for $\bar{p}_{b,s}$ can be found. By introducing Lagrange multiplier $\mu_{b}$ for $\mathbf{C_1}$, the Lagrange function is given by
\begin{equation}\label{eqn:Lagrange-Power}
    \begin{split}
       {\mathcal{L}_2}(\bar{\p},\boldsymbol{\mu}) =  {f}_2(\bar{\p},\boldsymbol{\gamma},\boldsymbol{Y}) - \Bsum \mu_b(\Ssum \bar{p}_{b,s}^2 - P_b^{\max}),
    \end{split}
\end{equation}
where $\boldsymbol{\mu}=\{\mu_1,\dots,\mu_B\}$.
The first-order optimality condition yields $\frac{\partial {\mathcal{L}_2}(\bar{\p},\boldsymbol{\mu})}{\partial \bar{p}_{b,s}}=0$, so we have
\begin{equation*}
   \hspace{-7.5cm} {\bar{p}_{b,s}}^{*} =
\end{equation*}
\vspace*{-0.7cm}
\begin{equation}\label{eqn:Power-Opt}
     \frac{\Nsum y_{b,s,n}\sqrt{\assocNBS w(1+\gam)\h}}{\Bpsum \Nsum \!\!q y^2_{\bp,s,n}\hbpbar\!\!\! +\!\!\!\! \Nsum y^2_{b,s,n}(\h \!\!\!+ \!\htilde)\!+\!\mu_b}.
\end{equation}
The Lagrange multiplier $\boldsymbol{\mu}$ satisfies the complementary slackness for the constraint $\mathbf{C_1}$, which ensures $\mu_b$ is either zero or a positive value that ensures the power budget constraint. It can be seen that ${\bar{p}_{b,s}}^{*}$ is a decreasing function of $\mu_b$. Thus, the values of $\boldsymbol{\mu}$ can be found by solving $J(\mu_b) = \Ssum {({\bar{p}_{b,s}(\mu_b)})}^2 - P_b^{\max} = 0$ using bisection search.
% \begin{equation}
%     J(\mu_b) = \Ssum {({\bar{p}_{b,s}(\mu_b)})}^2 - P_b^{\max} = 0
% \end{equation}
The overall procedure to obtain optimal power allocation  is summarized in \textbf{Algorithm-1}.

%%%%%%%%%%%%%%%%%%%%%%%%%%%%%%%%%%%
\begin{algorithm}[b!]
\small
    \caption{Centralized power allocation via fractional programming}
    \label{Algorithm-1}
    \textbf{Input}: Initializing $\bar{\p}$, $\boldsymbol{\gamma}$ with feasible values, maximum number of iterations $L_{\max}$, stopping accuracy $\epsilon_1$, bisection search parameters.
    \begin{algorithmic}[1]
        \For {$t=0,1,...$}
        \State Update $\boldsymbol{Y}^{(t+1)}$ using \eqref{eqn:Optimal-Y}
        \For {$\forall b \in \Bset$}
            \If{$J(\mu_b=0) \leq 0$}
                \State Update $\bar{p}_{b,s}^{(t+1)}$ using \eqref{eqn:Power-Opt} when $\mu_b = 0$
                \Else 
                \State Find $\mu_b^{*}$ using bisection search and update $\bar{p}_{b,s}^{(t+1)}$ using \eqref{eqn:Power-Opt}
                \EndIf
        \EndFor
        \State Update $\boldsymbol{\gamma}^{(t+1)}$ using \eqref{eqn:OptGamma}
        \State \textbf{Until} $|f(\bar{\p},\boldsymbol{Y},\boldsymbol{\gamma})^{(t+1)} - f(\bar{\p},\boldsymbol{Y},\boldsymbol{\gamma})^{(t)}| < \epsilon_1$ or $t=L_{\max}$
        \EndFor
    \end{algorithmic}
    \textbf{Output}: The optimal power allocation: $\p^{*} = {(\bar{\p}^{(t+1)})}^2$.
\end{algorithm}
%%%%%%%%%%%%%%%%%%%%%%%%%%%

\section{Low-Complexity Distributed Solution and Complexity Analysis}
In this section, we provide a light-weight distributed power control solutions for time-sensitive wireless applications and analyze the worst-case complexity of the proposed algorithms.

\subsection{Low-Complexity Distributed Solution}
The optimization problem $\pr_{5}$ is a quadratically constrained quadratic program (QCQP) that can be efficiently solved using the ADMM method in a distributed manner\cite{ADMM-ML}. We first reformulate  $\pr_{5}$ by introducing first and second blocks of primal variables. We then proceed to employ the scaled ADMM approach to achieve power allocation in a distributed manner with reduced complexity.  
% \begin{mini!}[3]
% 	{\bar{\p},\boldsymbol{\Delta}} {-\tilde{f}_2(\A,\bar{\p},\boldsymbol{\gamma},\boldsymbol{Y}) + \mathbbm{1}_f(\boldsymbol{\Delta})\label{Con:p6Obj}}
% 	{\label{problem:p6main}}{\pr_{6}: \ \ }
%    \addConstraint{\boldsymbol{\Delta}=\bar{\p},}{}\label{Con:p6Primal}
% \end{mini!}
By introducing the auxiliary variables $\boldsymbol{\Delta}=[\boldsymbol{\delta}_{1},\dots,\boldsymbol{\delta}_B]^T \in \mathbb{R}_+^{B\times S}$ and $\boldsymbol{\delta_b}=[\delta_{b,1},\dots,\delta_{b,S}]^T$, such that $\delta_{b,s}$ is equivalent to $\bar{p}_{b,s}$, problem $\pr_5$ is given as:
\begin{eqnarray}
\label{prob6} \pr_6:\hspace*{-4mm}
&&\hspace*{2mm}\underset{\bar{\p},\boldsymbol{\Delta}}{\textrm{minimize}} \,\, \,\, -{f}_2(\bar{\p},\boldsymbol{\gamma},\boldsymbol{Y}) + \mathbbm{1}_{\mathbf{C_1}}(\boldsymbol{\Delta})\\
\mbox{s.t.}\hspace*{-4mm}
&&\mathbf{C_{10}}:\boldsymbol{\Delta}=\bar{\p}\notag,
\end{eqnarray}
In order to take the constraint $\mathbf{C_1}$ into account, for fixed $\boldsymbol{Y}$, and $\boldsymbol{\gamma}$, we define the following indicator function, which is used in the objective function of $\pr_6$.
\begin{equation}\label{eqn:p5ConstraintSet}
\mathbbm{1}_{\mathbf{C_1}}(\boldsymbol{\Delta})= 
\left\{\begin{array}{ll}
 0 &   \Ssum \delta_{b,s}^2 \leq P_b^{\max}, \forall b \in \Bset\\
 \infty & \text { o/w }.
\end{array}\right.
\end{equation}
According to the scaled ADMM method, we consider $\bar{\p}$ and $\boldsymbol{\Delta}$ as the first and second blocks of primal variables \cite{ADMM-ML}. Then, we introduce $\boldsymbol{Z}=[\boldsymbol{z}_1,\dots,\boldsymbol{z}_B]^T \in \mathbb{R}_+^{B\times S}$, and $\boldsymbol{z_b}=[z_{b,1},\dots,z_{b,S}]^T$ as the dual variables, accounting for the equality constraint in $\mathbf{C_{10}}$. The augmented Lagrangian function for optimization problem $\pr_6$ is obtained as follows:
\begin{equation*}
\hspace{-6cm}\mathcal{L}_D(\bar{\p},\boldsymbol{\Delta},\boldsymbol{Z}) =
\end{equation*}
\vspace{-0.4cm}
\begin{equation}\label{eqn:P6Lagrang}
 -{f}_2(\bar{\p},\boldsymbol{\gamma},\boldsymbol{Y}) + \mathbbm{1}_{\mathbf{C_1}}(\boldsymbol{\Delta}) + \frac{\rho}{2}\Bsum {\lVert \boldsymbol{\delta}_b - \bar{\boldsymbol{p}}_b + \boldsymbol{z}_b \rVert}^2_2,
\end{equation}
where $\bar{\boldsymbol{p}}_b = [\bar{p}_{b,1},\dots,\bar{p}_{b,S}]^T$, and $\rho > 0$ is the penalty factor of the augmented Lagrangian. The augmented Lagrangian maximization and multiplier update can be conducted iteratively through the ADMM approach. 
% First, the central unit processor (CPU), to which all BSs are connected through backhaul links, solves the following optimization problem. for given $\boldsymbol{Y}$
In what follows, we first explain how to obtain $\bar{\p}$, $\boldsymbol{\Delta}$, and $\boldsymbol{Z}$, and then describe how the process is distributed over all BSs.

\subsubsection{Obtaining the first block of primal variables, $\bar{\p}$}
The power allocation matrix, $\bar{\p}$ is derived by solving the following optimization problem:
\begin{equation*}
    \hspace{-4cm}\bar{\p} \leftarrow \arg \underset{\bar{\p}}{\min} \ \mathcal{L}_D(\bar{\p},\boldsymbol{\Delta},\boldsymbol{Z}) =
\end{equation*}
\begin{equation}\label{eqn:PowerADMM-Opt}
     \arg \underset{\bar{\p}}{\min} -{f}_2(\bar{\p},\boldsymbol{\gamma},\boldsymbol{Y}) + \frac{\rho}{2}\Bsum {\lVert \boldsymbol{\delta}_b - \bar{\boldsymbol{p}}_b + \boldsymbol{z}_b \rVert}^2_2.
\end{equation}
The closed-form solution to the problem in \eqref{eqn:PowerADMM-Opt} is derived by setting $\frac{\partial \mathcal{L}_D(\bar{\p},\boldsymbol{\Delta},\boldsymbol{Z})}{\partial \bar{p}_{b,s}} = 0$.
\begin{equation*}
   \hspace{-7.5cm} \bar{p}_{b,s} =
\end{equation*}
\vspace{-0.6cm}
\begin{equation}\label{eqn:PowerADMM}
     \frac{\Nsum y_{b,s,n}\sqrt{\assocNBS w(1+\gam)\h}\!+\frac{\rho}{2}(\delta_{b,s}+z_{b,s})}{\Bpsum \Nsum \!\!q y^2_{\bp,s,n}\hbpbar\!\!\! +\!\!\! \Nsum y^2_{b,s,n}(\h\!\!\! +\!\! \htilde)\!+\!\frac{\rho}{2}}.
\end{equation}

\subsubsection{Obtaining the second block of primal variables, $\boldsymbol{\Delta}$}
In order to obtain the optimal second block of primal variables, i.e., $\boldsymbol{\Delta}$, we need to solve the following optimization problem:
\begin{equation*}
    \hspace{-4cm}\boldsymbol{\Delta} \leftarrow \arg \underset{\boldsymbol{\Delta}}{\min} \ \mathcal{L}_D(\bar{\p},\boldsymbol{\Delta},\boldsymbol{Z}) =
\end{equation*}
\begin{equation}\label{eqn:DeltaADMM-Opt}
     \arg \underset{\boldsymbol{\Delta}}{\min} \ \mathbbm{1}_{\mathbf{C_1}}(\boldsymbol{\Delta}) + \frac{\rho}{2}\Bsum {\lVert \boldsymbol{\delta}_b - \bar{\boldsymbol{p}}_b + \boldsymbol{z}_b \rVert}^2_2.
\end{equation}
The optimization problem in \eqref{eqn:DeltaADMM-Opt} can be decomposed into $B$ sub-problems that can be solved in parallel. The $b$-th sub-problem is given as:
% \begin{mini!}[3]
% 	{\boldsymbol{\delta}_b} { \frac{\rho}{2}\Ssum {(\delta_{b,s}-\bar{p}_{b,s}+z_{b,s})}^2\label{Con:p7Obj}}
% 	{\label{problem:p7main}}{\pr_{7}: \ \ }
%    \addConstraint{\Ssum \delta_{b,s}^2 \leq P_{b}^{\max}}{}\label{Con:p7Power}
% \end{mini!}
\begin{eqnarray}
\label{prob7} \pr_7:\hspace*{-4mm}
&&\hspace*{2mm}\underset{\boldsymbol{\delta}_b}{\textrm{minimize}} \,\, \,\, -\frac{\rho}{2}\Ssum {(\delta_{b,s}-\bar{p}_{b,s}+z_{b,s})}^2\\
\mbox{s.t.}\hspace*{-4mm}
&&\mathbf{C_{13}}:\Ssum \delta_{b,s}^2 \leq P_{b}^{\max}\notag,
\end{eqnarray}
Each sub-problem in $\pr_7$ is a projection onto $l2$ ball (a convex problem) that can be solved by utilizing stationarity and complementary slackness conditions \cite{boyd2004convex}. Therefore, the optimal solution of the second block of primal variables is obtained as:
\begin{equation}\label{eqn:OptDelta}
    \boldsymbol{\delta}_b =\boldsymbol{\xi}_b \times \min\left\{\frac{\sqrt{P_b^{\max}}}{{\lVert \boldsymbol{\xi}_b \rVert}_2},1\right\}, \ \forall b \in \Bset,
\end{equation}
where $\boldsymbol{\xi}_b = [\xi_{b,1},\dots,\xi_{b,S}]^T= [\bar{p}_{b,1}-z_{b,1},\dots,\bar{p}_{b,S}-z_{b,S}]^T$. Similarly, the dual variable $\boldsymbol{Z}$ can be obtained by $\bar{\p}-\boldsymbol{\Delta}$.
In \textbf{Algorithm-2}, the overall process of power allocation with ADMM combined with FP is summarized. First, the central processor unit (CPU) computes the auxiliary variables  $\boldsymbol{Y}$, $\boldsymbol{\gamma}$ and then the first block of primal variables, i.e., $\bar{\p}$. Then, the CPU transmits $\bar{\p}$ to all BSs, and each BS obtains the optimal value of the second bock of primal variables, i.e., $\boldsymbol{\Delta}$ and sends them back to the CPU. This process continues until convergence or reaches the maximum number of iterations. It is noteworthy that there is a trade-off between the number of iterations in the ADMM, which is the inner algorithm in \textbf{Algorithm-2}, and the optimality of the solution. Therefore, for faster convergence, we need to adjust the penalty factor $\rho$. In \textbf{Algorithm-2}, BSs are not required to possess the CSI of every user associated with other BSs. They need only possess a copy of $\boldsymbol{\xi}_b$. If we wish to implement \textbf{Algorithm-1} in the same distributed manner, each BS must have the CSI of all users, thereby increasing the communication overhead. In addition, \textbf{Algorithm-1} is computationally more expensive when obtaining $\bar{\p}$ at each BS.

\begin{algorithm}[b!]
\small
    \caption{ADMM-based distributed power allocation}
    \label{Algorithm-2}
    \textbf{Input}: Initializing $\bar{\p}$, $\boldsymbol{\gamma}$ with feasible values, maximum number of iterations $L_{\max}$ for FP, and $L_{\max}^A$ for ADMM, stopping accuracy $\epsilon_1$ for FP and $\epsilon_A$ for ADMM, and $\rho$.
    
    \begin{algorithmic}[1]
        \For {$t=0,1,...$}
            \State Update $\boldsymbol{Y}^{(t+1)}$ using \eqref{eqn:Optimal-Y}
            \State Initialize $\boldsymbol{\Delta},\boldsymbol{Z}$
            \For {$k=0,1,\dots$}
                \State Update $\bar{\p}^{(k+1)}$ using \eqref{eqn:PowerADMM} and send it to all BSs
                \State Each BS updates $\boldsymbol{\delta}_b$ using \eqref{eqn:OptDelta} and sends it back to the CPU
                \State Update $\boldsymbol{Z}^{(k+1)}$ using $\boldsymbol{Z}^{(k+1)} \leftarrow \boldsymbol{Z}^{(k)} + \bar{\p}^{(k)}-\boldsymbol{\Delta}^{(t)}$
                \State \textbf{Until}
                $\frac{{\lVert \bar{\p} -\boldsymbol{\Delta} \rVert}_2}{{\lVert \boldsymbol{\Delta} \rVert}_2} < \epsilon_A$ or $k = L_{\max}^A$
                \EndFor
            \State Update $\boldsymbol{\gamma}^{(t+1)}$ with $\boldsymbol{\Delta}^{(k+1)}$ using \eqref{eqn:OptGamma}
        \State \textbf{Until} $|f(\bar{\p},\boldsymbol{Y},\boldsymbol{\gamma})^{(t+1)} - f(\bar{\p},\boldsymbol{Y},\boldsymbol{\gamma})^{(t)}| < \epsilon_1$ or $t=L_{\max}$
        \EndFor
    \end{algorithmic}
    \textbf{Output}: The optimal power allocation: $\p^{*} = {(\bar{\p}^{(t+1)})}^2 = {(\boldsymbol{\Delta}^{(t+1)})}^2$.
\end{algorithm}
\begin{algorithm}[b!]
\small
    \caption{Overall two-stage algorithm for solving $\pr_1$}
    \label{Algorithm-3}
    \textbf{Input:} Upper bound for fractional bandwidth, i.e., $B_{\mathrm{th}}$, start and end frequencies of the TW of interest, $f_I$ and $f_E$, tolerance value of the approximated window, $\epsilon$.
    \\
    \textbf{Stage 1:} 
    \begin{algorithmic}[1]
        \State Obtain the feasible values of $\hat{w}_I$ and $\hat{w}_E$ using interior-point method.
        \State Obtain $S_{\mathrm{LB}}$ using \eqref{eqn:LowerBound-S}
    \end{algorithmic}
    \textbf{Stage 2:}
    \begin{algorithmic}[1]
        \State Initialize a feasible $\A^{(0)}$ heuristically.
        Using $\A^{(0)}$, initialize $\p$ with equal power allocation and a scale factor between 0 and 1.
        \For {$i=0,1,...$}
            \State Solve $\pr_{3}$ using the dual-simplex method and obtain $\A^{(i)}$.
            \State Solve power allocation sub-problem using either \textbf{Algorithm-1} or \textbf{Algorithm-2} and obtain $\p^{(i)}$
            \State \textbf{Until} $|{f(\p,\A)}^{(i+1)} - {f(\p,\A)}^{(i)}| < \epsilon_3$
        \EndFor
    \end{algorithmic}
    \textbf{Output}: The optimal power allocation, $\p^{*}$, Joint user association and sub-band assignment, $\A^*$, Number of sub-bands $\hat{S}$, and starting and end edge bands, $\hat{w}_I, \hat{w}_E$.
\end{algorithm}
%%%%%%%%%%%%%%%%%%%%%%%%%%%
\subsection{Complexity Analysis of the Proposed Algorithms}
In what follows, we provide a worst-case complexity analysis of the proposed algorithms. 
% It is worth mentioning that the upper bound expressions are considered for complexity since the joint user association and sub-band assignment can reduce the number of non-zero elements in each variable. 
\subsubsection{\textbf{Algorithm-1}} The complexity of {Algorithm~1} is calculated as follows. Step-2 and Step-10 of {\textbf{Algorithm-1}}, each has the complexity of $\mathcal{O}(BSN)$. The bisection search in Step-7 requires $\mathcal{O}(B\log(\frac{\mu_{\max}}{\epsilon_b}))$, where $\mu_{\max}$ is the upper value of the Lagrange multiplier, $\mu_b$, and $\epsilon_b$ is the error tolerance of bisection search. Step-7 also has $\mathcal{O}(BS)$ for calculating the power allocation. Therefore, the overall complexity of \textbf{Algorithm-1} is $\mathcal{O}(I_P(2BSN+BS+B\log(\frac{\mu_{\max}}{\epsilon_b})))$, such that $I_P$ is the number of iterations \textbf{Algorithm-1} requires to converge based on $\epsilon_1$.
\subsubsection{\textbf{Algorithm-2}} For the inner algorithm (i.e., ADMM), Step-5 and Step-6 require the complexity of $\mathcal{O}(BS)$ together. Similar to \textbf{Algorithm-1}, the complexity of Step-2 and Step-10 is $\mathcal{O}(BSN)$. Hence, the overall complexity of Algorithm-2 is obtained as $\mathcal{O}(I_{FP}(2BSN+I_{AD}2BS))$, where $I_{FP}$ and $I_{AD}$ is the number of iterations required for convergence in the outer and inner ADMM algorithm, respectively. Based on our experiments, we observed that by selecting  proper penalty factor for ADMM, i.e., $\rho$, the ADMM algorithm can converge in a few iterations. Therefore, the overall complexity of \textbf{Algorithm-2} is less than that of Algorithm~1. 

\subsubsection{\textbf{Algorithm-3}}
The complexity of Stage 1 in \textbf{Algorithm-3} is $\mathcal{O}(8I_I)$, such that $I_I$ is the number of iterations for the interior-point method. In Stage 2, the order of complexity in the joint user association and sub-band assignment when solving LP in $\pr_3$ is $\mathcal{O}({(BSN)}^3)$. Solving $\pr_3$ without using \textbf{Theorem 1} results in the complexity of $\mathcal{O}(2^{BSN})$. As a result, the overall complexity of solving $\pr_1$ with the solution presented in \textbf{Algorithm-3} is $\mathcal{O}(8I_I + I_M({(BSN)}^3 + I_P(2BSN+BS+B\log(\frac{\mu_{\max}}{\epsilon_b}))))$ if \textbf{Algorithm-1} is used for power allocation and is $\mathcal{O}(8I_I + I_M( {(BSN)}^3 + I_{FP}(2BSN+I_{AD}2BS)))$ if \textbf{Algorithm-2} is utilized for solving $\pr_4$, where $I_M$ is the number of iterations required for \textbf{Algorithm-3} to converge. { The proof of convergence of the alternating optimization employed in the second stage of \textbf{Algorithm-3} to a sub-optimal point is provided in Appendix B.}

\section{Numerical Results and Discussions}
\begin{figure}[t]
    \centering
\includegraphics[scale=0.5]{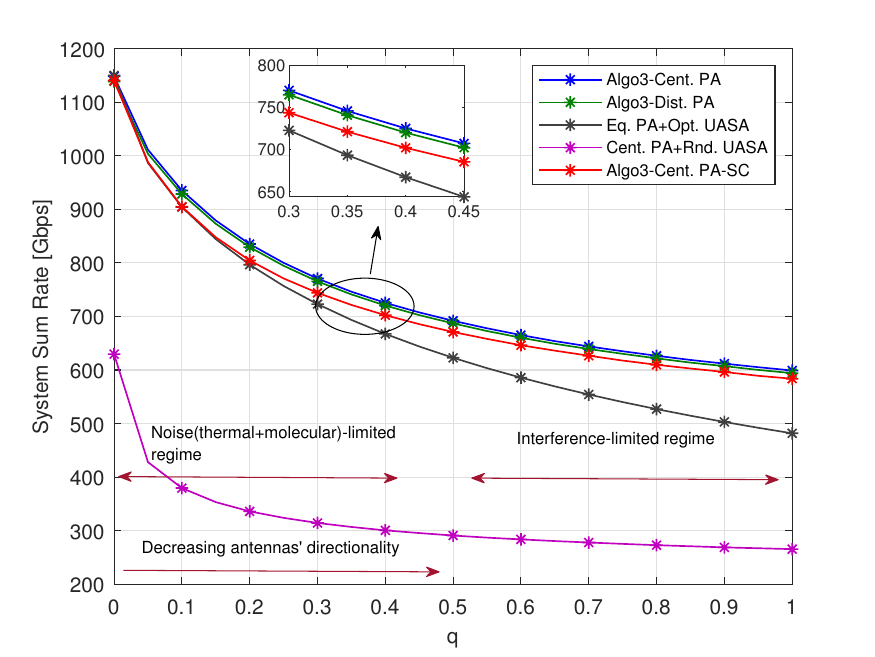}
    \vfill
{\includegraphics[scale=0.5]{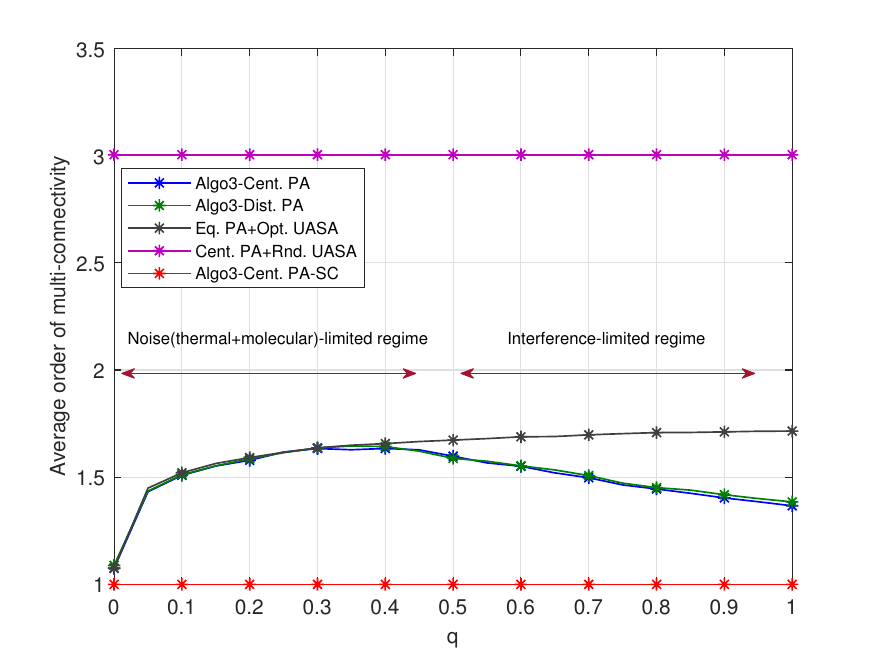}
    \label{fig:q_Changes_MultiConnOrder}}
    \caption{Varying antennas' directionality on (a) network sum-rate (b) average multi-connectivity order. $\Gamma^L_n=1$, $B=6$, $N=12$.}
    \label{fig:q_changes}
    \vspace{-5mm}
\end{figure}
This section focuses on evaluating the performance of the proposed algorithms and presenting their convergence behaviors. The results extract useful insights related to the impact of the number of BSs, blockages, molecular absorption, antennas' directionality, {hardware impairment, multi-antenna BSs, and imperfect CSI} on the performance. 
\subsubsection{Simulation Set-up and Parameters}
Unless otherwise stated, 
{we consider that the users and BSs are randomly distributed across a circular region} with a radius of 30 [m], the maximum transmit power is set to be $P_b^{\max}= 1$ Watts, the shape parameter of Nakagami-$m$ fading is $m=20$, fractional bandwidth is $B_{\mathrm{th}}=1\%$, array gains are $G_{\max}^{\text{Tx}}=G_{\max}^{\text{Rx}}=25$ [dB]. {The coverage of each BS is determined by transmit power and antenna gains.} The blockage density is $\eta = 0.005 \ [{\textrm{m}}^{-1}]$, the penalty factor of ADMM is $\rho=2.2$, and the accuracy of all three algorithms is set to be $\{\epsilon_1,\epsilon_2,\epsilon_3\} = 10^{-3}$. 
{In this section, the considered gas molecules used to calculate the molecular absorption coefficient and their respective ratios are as follows:
    N2 (Nitrogen): 76.545\% - O2 (Oxygen): 20.946 \% - H2O (Water): 1.57\% - CO2 (Carbon dioxide): 0.033 \% - CH4 (Methane): 0.906\%}

In the figures of this section, \textit{`Algo3 - Cent. PA'} and \textit{`Algo3 - Dist. PA'} represent the use of Algorithm-3 with \textbf{Algorithm-1} and \textbf{Algorithm-2}, respectively. \textit{`Algo3 - Cent. PA - SC'} stands for the case of using single connectivity, in which we constrain the users to associate with only one BS.
\textit{`Eq. PA + Opt. UASA'} denotes the use of equal power allocation and the proposed joint user association and sub-band assignment, and \textit{`Cent. PA + Rnd. UASA'} shows the use of proposed power allocation with \textbf{Algorithm-1} along with a feasible random user association and sub-band assignment. Moreover, \textit{`Average order of multi-connectivity'} (AOM) is obtained by counting the number of BSs to which each user is associated and then averaging out over all users.

\begin{figure*}[t!]
\begin{minipage}{0.6\textwidth}
{\includegraphics[scale=0.38]{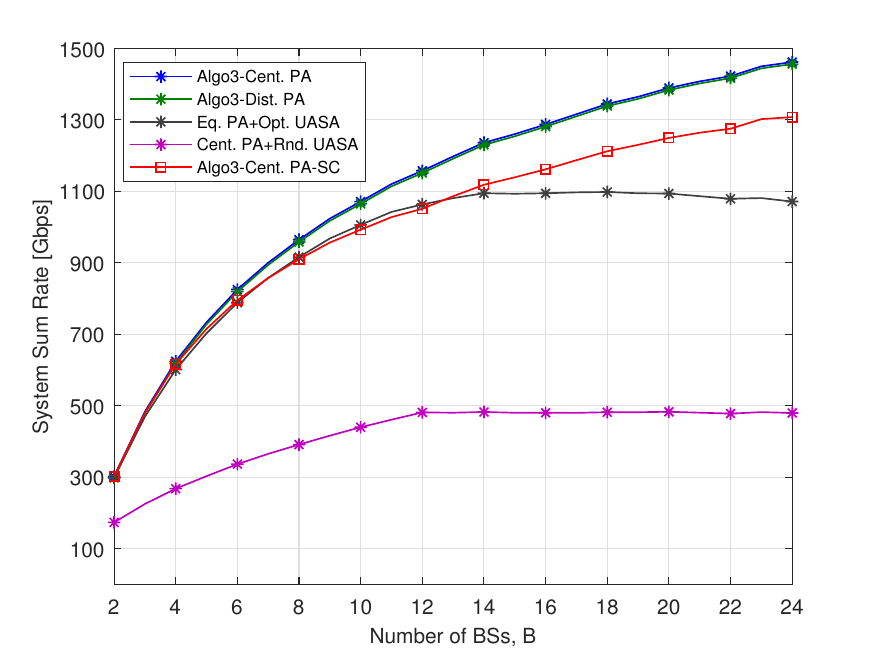} 
}
{\includegraphics[scale=0.38]{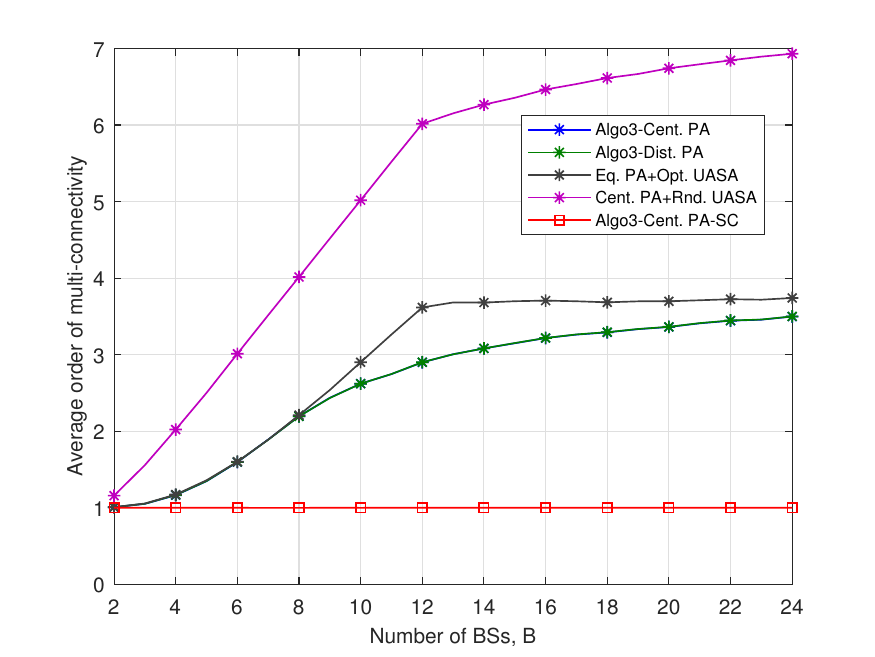}
    \label{fig:NumBSs_Varies}}
     \caption{(a) System sum-rate (b) AOM vs number of BSs, $\Gamma^L_n=1$, $q=0.2$, $N=12$.}
      \label{fig:NumBSs_Varies}
     \end{minipage}\hfill
     \begin{minipage}{0.32\textwidth}
     \includegraphics[scale=0.38]{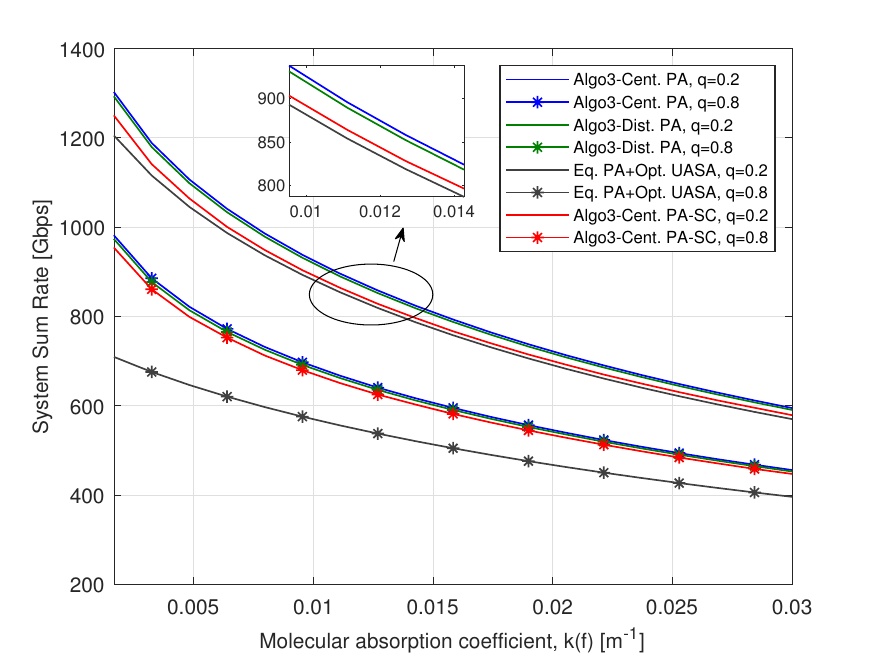}
    \caption{Network sum-rate vs molecular absorption coefficient, $\Gamma^L_n=1$, $B=6$, $N=12$.}
    \label{fig:Kf_Varies}
    \end{minipage}
\end{figure*}
%%%%%%%%%%%%%%%%%%%%%%%%%%%%%%%%%%%%%
\subsubsection{Impact of Antannas' Directionality}

Fig.~\ref{fig:q_changes} depicts the impact of altering the antennas' directionality $q$ on the network sum-rate and AOM. We note that the centralized and distributed versions of \textbf{Algorithm-3} outperform \textit{`Eq. PA + Opt. UASA'} and \textit{`Cent. PA + Rnd. UASA'}, especially for scenarios  when $q$ is high, i.e., low directionality resulting in high interference; thereby interference-limited regime. Subsequently, increasing $q$ (lowering antennas' directionality) degrades the performance of all methods. Moreover, it can be observed that multi-connectivity improves the system sum-rate, especially  for high directionality. Fig. \ref{fig:q_changes}b shows that in the noise-limited regime, AOM increases as $q$ increases. However, in the interference-limited regime, as opposed to \textit{`Eq. PA + Opt. UASA'} the AOM of the proposed methods decreases due to optimal allocations decrease the impact of interference on the system sum-rate.

% The reason for having almost similar performance in smaller values of $q$ is that the less is $q$, the lower the impact of interference in the achievable rate. Hence, even equal power allocation can demonstrate a comparable performance with respect to the proposed power allocation methods. Both proposed methods, however, outperform Eq-Power, and their degradation is less than that of Eq-Power as shown in the figure. 

\subsubsection{Impact of the Number of BSs}

Fig.~\ref{fig:NumBSs_Varies} illustrates the influence of increasing the number of BSs on the system performance for both proposed methods and other benchmarks. It is evident that as the number of BSs increases, both \textit{`Algo3 - Cent. PA'} and \textit{`Algo3 - Dist. PA'} outperform \textit{`Eq. PA - Opt. UASA'} and\textit{`Cent. PA + Rnd. UASA'}. Also, it is evident that the proposed single-connectivity case shows an inferior performance compared to the proposed methods with multi-connectivity, and the gap increases with the number of BSs.
The improvement in the sum-rate achieved by random user association and sub-band assignment reaches a saturation point, highlighting the superior performance of the proposed low-complexity joint user association and sub-band assignment. Moreover, the gain exhibited by the proposed methods over \textit{`Eq. PA - Opt. UASA'}, which benefits from the proposed joint user association and sub-band assignment, increases with the number of BSs. This improvement is attributed to the optimal power allocation that compensates for increased interference levels while effectively leveraging multi-connectivity. Furthermore, from Fig.~\ref{fig:NumBSs_Varies}b, we can observe that the AOM of the proposed methods increases with $B$, which reflects the impact of leveraging multi-connectivity in such as system. Although the increment of AOM is seen for \textit{`Eq. PA - Opt. UASA'}, due to the lack of optimal power allocation along with the user association, the increase in AOM cannot compensate for the increased level of interference as the total number of BSs rises.

\subsubsection{Effect of Molecular Absorption Coefficient}

{Varying the molecular absorption coefficient directly, for a given carrier frequency, can reflect different environmental factors, such as pressure, temperature, altitude, and so on \cite{MBN-Survey}.} Fig.~\ref{fig:Kf_Varies} highlights that the system sum-rate decreases when the molecular absorption coefficient $k(f)$ increases. The proposed methods, however, outperform the \textit{`Eq. PA - Opt. UASA'}, especially for smaller values of $k(f)$ and reduced antennas' directionality, i.e., $q=0.8$, as they can optimally reduce the interference for a given user. Moreover, raising $k(f)$ reduces the differences between the proposed methods and \textit{`Eq. PA - Opt. UASA'}, even for $q=0.8$. The reason is that since $k(f)$ impacts the channels exponentially, the channel gain of the signal power in the numerator of the SINR for each user decreases as the molecular absorption coefficient increases. It also increases the impact of molecular absorption noise. Therefore, although still outperforming \textit{`Eq. PA - Opt. UASA'}, the proposed algorithms' performance is impeded by channel gains when having an increased molecular absorption coefficient. %Again Fig. 4(b) depicts that the proposed algorithms strikes the balance between the multi-connectivity required  with \textit{`Eq. PA - Opt. UASA'} and single-connectivity to optimize the system performance gains. 

\subsubsection{Blockage Density}
\begin{figure}[h]
    \centering
{\includegraphics[scale=0.5]{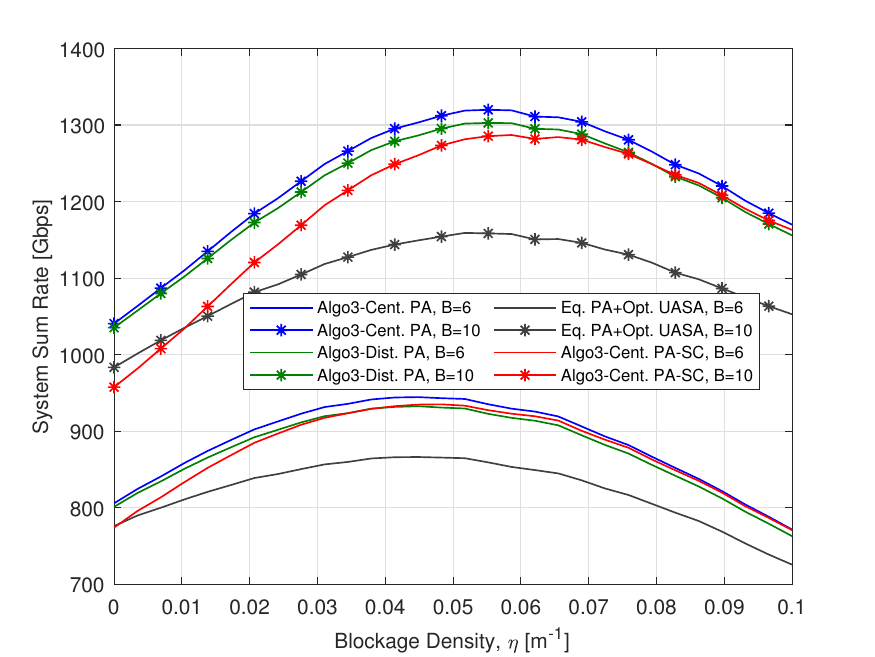}
    \label{fig:Blockage_Varies_SumRate}}
    \vfill
{\includegraphics[scale=0.5]{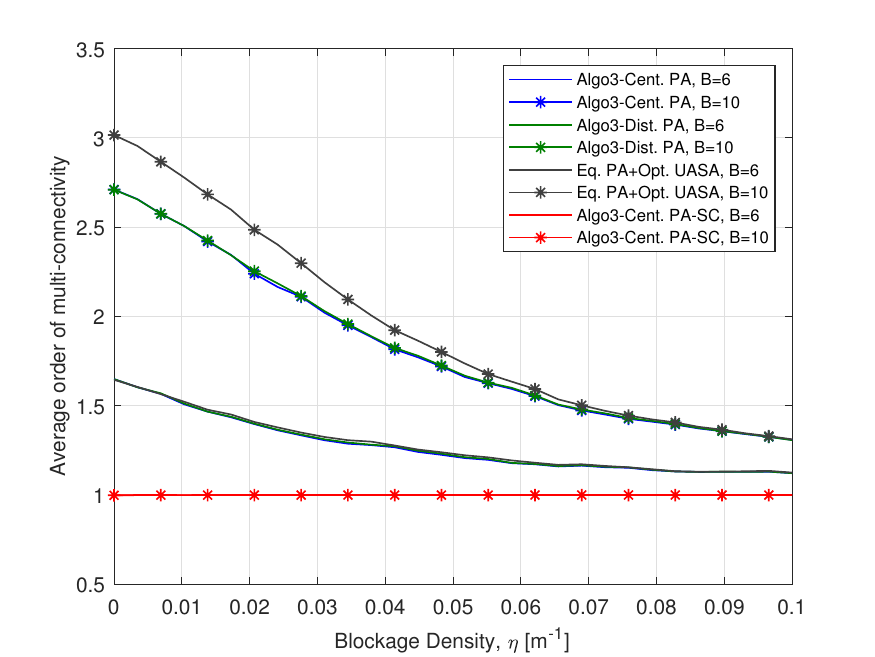}
    \label{fig:Blockage_Varies_MultiConnOrder}}
    \caption{ (a) System sum-rate and (b) average multi-connectivity order vs blockage density, $\Gamma^L_n=1$, $q=0.2$, $N=12$.}
    \label{fig:Blockage_Varies}
    %\vspace{-5mm}
\end{figure}

\begin{figure*}[t!]
\begin{minipage}{0.48\textwidth}
\centering
    \includegraphics[scale=0.5]{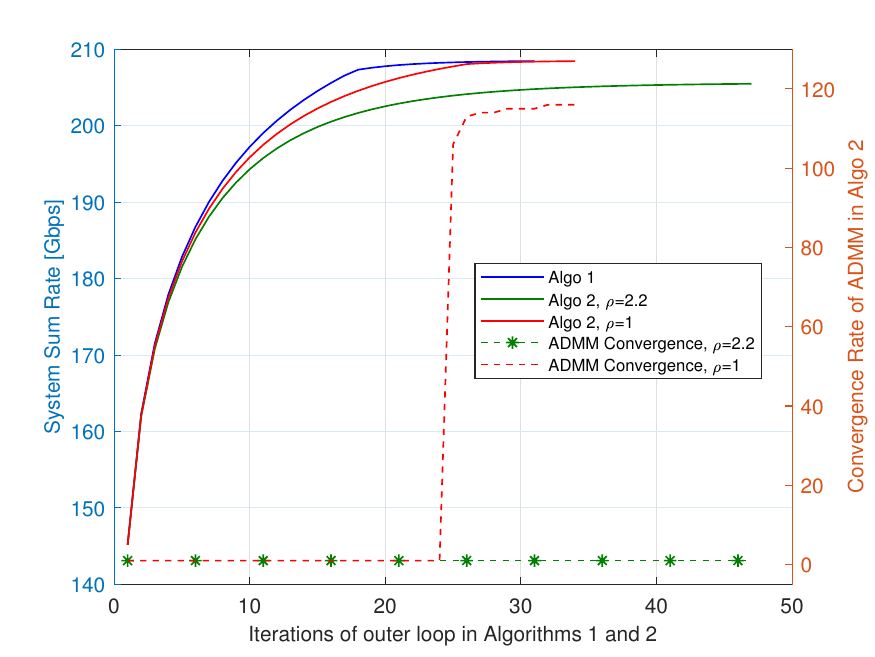}
    \caption{Convergence of the proposed algorithms for power allocation, and the impact of varying penalty factor in Algorithm-2.}
    \label{fig:Convergence_PA}
\end{minipage}
\hfill
\begin{minipage}{0.48\textwidth}
\centering
    \includegraphics[scale=0.5]{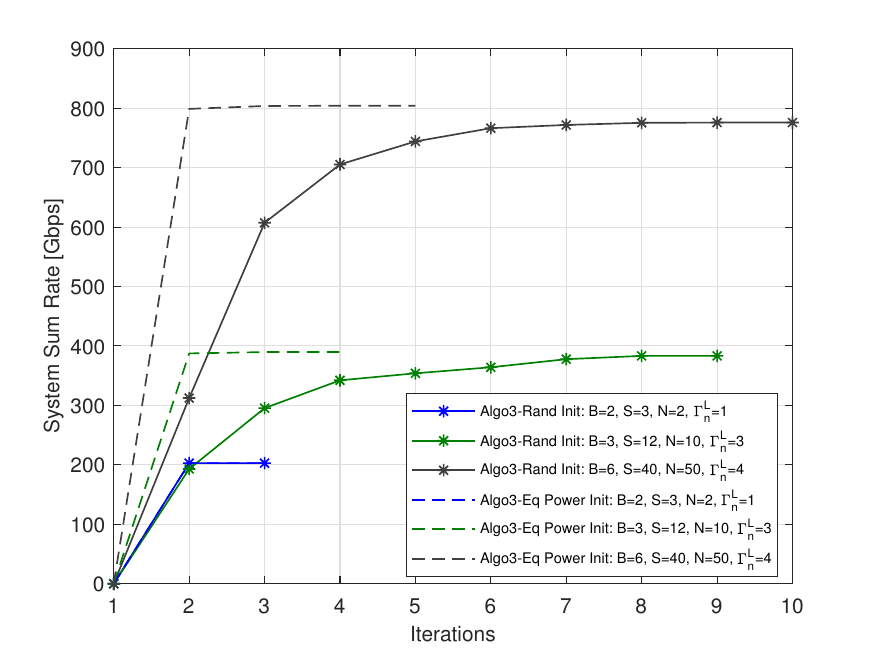}
    \caption{Convergence of Algorithm-3 for different setups, and minimum  number of sub-band per user, i.e., $\Gamma^L_n$, and different initializations}
\label{fig:Convergence_MainAlgo}
\end{minipage}
    \vspace{-5mm}
\end{figure*}

Fig.~\ref{fig:Blockage_Varies} illustrates the impact of increasing blockage density ($\eta$) on the system sum-rate for different numbers of BSs. %Before discussing Figure \ref{fig:Blockage_Density_Changes}, it is important to clarify how blockage affects the system differently based on the directionality of the antennas. When employing antennas with lower directionality, which is more practical, the interference has a greater impact on the system sum-rate (indicated by increasing $q$). This allows interference to have a stronger influence on the proposed power allocation and user association. Consequently, when $q$ is higher, the absence of blocked links in the interference becomes more significant for a given user. As the proposed power allocation and user association aim to maximize the system sum-rate, more blocked links are assigned to the interference of that particular user. Therefore, the sum-rate improves as the blockage density increases. However, this increase is not observed when the impact of interference is constrained by a smaller value of $q$.
Observing Fig.~\ref{fig:Blockage_Varies}a, it becomes evident that as the blockage density increases, both proposed methods  outperform \textit{`Eq. PA - Opt. UASA'}. For lower values of $\eta$, the gap between multi-connectivity and single-connectivity widens. This is due to the availability of more unblocked links for connecting to multiple BSs. Furthermore, it is noticeable that employing a greater number of BSs enhances performance and increases the performance gap between the proposed methods and other benchmarks.

When blockage density increases, the proposed power allocation and user association strategies aim to maximize the system sum-rate. Consequently, more blocked links are observed at the interference of a given user. Therefore, the system sum-rate improves as blockage density rises. However, after reaching a certain blockage density, the overall data rate starts to decline. This decline is a consequence of an increase in $\eta$, which exponentially raises the likelihood of blocking the optimal link for a given user. Fig.~\ref{fig:Blockage_Varies}b illustrates that the AOM for all cases decreases as $\eta$ increases, i.e., in strong blockage environments, the proposed solution discourages multi-connectivity. Nonetheless, employing a greater number of BSs encourages the AOM level.

\subsubsection{Convergence Analysis}

Fig.~\ref{fig:Convergence_PA} illustrates the convergence of Algorithms 1 and 2, which are represented in the figure by \textit{`Algo 1'} and \textit{`Algo 2'}, respectively. 
It is evident that the proposed algorithms for power allocation converge in a few iterations. Also, the convergence of \textbf{Algorithm-2} is depicted for different penalty factors of the ADMM algorithm, i.e., $\rho$ to reflect the impact of the ADMM algorithm convergence on the overall \textbf{Algorithm-2}. The number of iterations at which the inner ADMM algorithm converges is displayed on the right y-axis. It can be seen that by reducing $\rho$, \textbf{Algorithm-2} converges faster and approaches \textbf{Algorithm-1} at the expense of having a greater convergence rate for the inner ADMM algorithm. Using $\rho =2.2$ for the penalty factor lets the inner algorithm converge in only one iteration, which can significantly reduce the communication overhead in the system. It can be deduced that there is a trade-off between the optimality of the solution obtained by \textbf{Algorithm-2} and the overall complexity  when \textbf{Algorithm-2} is utilized in a distributed manner.

Fig.~\ref{fig:Convergence_MainAlgo} shows the convergence of the overall proposed \textbf{Algorithm-3} for various system configurations, user association and sub-band assignment initialization, and the minimum required number of sub-bands for users, i.e., $\Gamma^L_n$. In this figure, \textit{`Algo 3-Rand Init'} stands for the case where the initialization for the user association and sub-band assignment, i.e., $\A$ is a random and feasible point. \textit{`Algo 3-Eq Power Init'} is the case when the initial $\A$ is obtained by calculating the rates using equal power allocation. It is evident that increasing the number of users, BSs, and sub-bands increases the convergence rate of Algorithm-3. The reason is that the total number of optimization variables is increases, leading to expanding the feasible set and search space of the optimization problem that can take more time to converge. Additionally, as shown in the previous section, the computational complexity of \textbf{Algorithm-3} increases as a function of optimization variables, which in turn, requires the algorithm a greater number of iterations to converge. It can also be seen that with random initialization for user association and sub-band assignment, Algorithm-3 can achieve a comparable performance compared to the case with equal power allocation-based initialization.

\begin{figure}[h]
    \centering
{\includegraphics[scale=0.5]{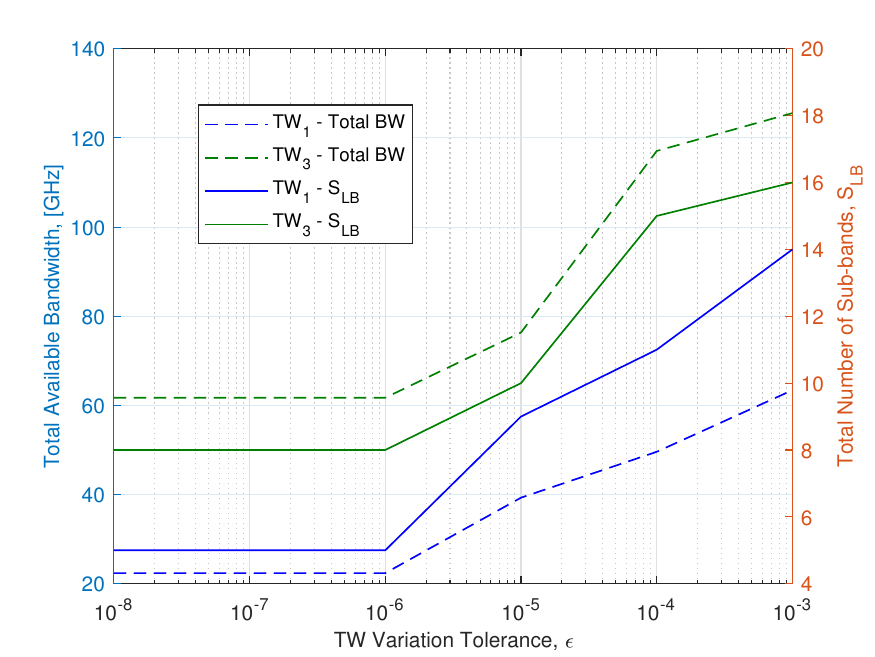}
    \label{fig:DiffTWs}}
    \vfill
{\includegraphics[scale=0.5]{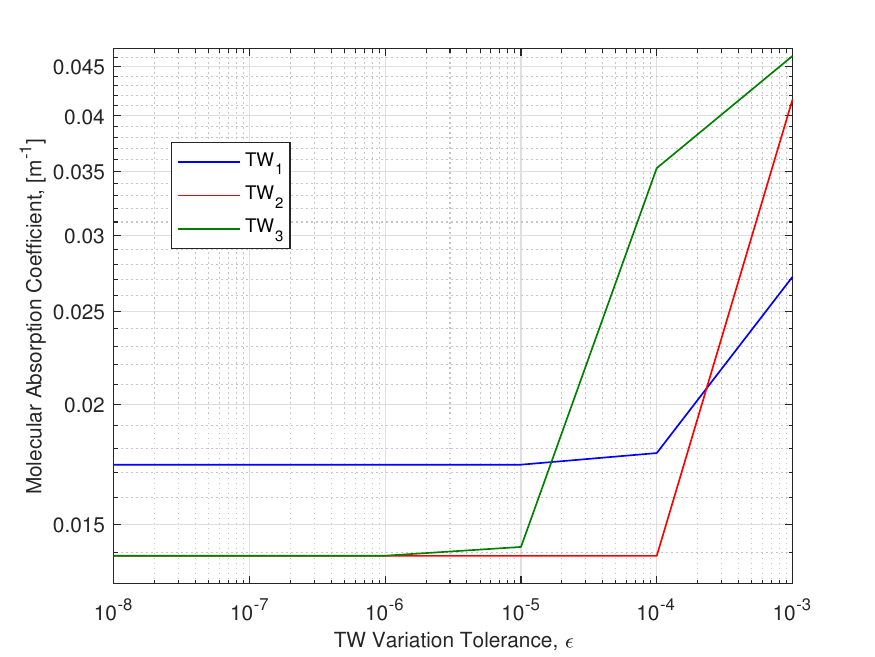}
    \label{fig:DiffTWs-MolecAbs}}
    \caption{{ (a) Total available bandwidth and optimal number of sub-bands (b) molecular absorption coefficient $k(f)$, for different TWs versus molecular absorption coefficient variation tolerance, $\epsilon$.}}
    \label{fig:DifTWs-BW-SubBands-MolecAbs}
    \vspace{-5mm}
\end{figure}

\subsubsection{{Number of Sub-bands, Available Bandwidth and Molecular Absorption Coefficient}}
%%%%%%%%%%%%%%%%
\begin{figure*}
    \centering
    \begin{minipage}{0.32\textwidth}
        \centering
        \includegraphics[scale=0.43]{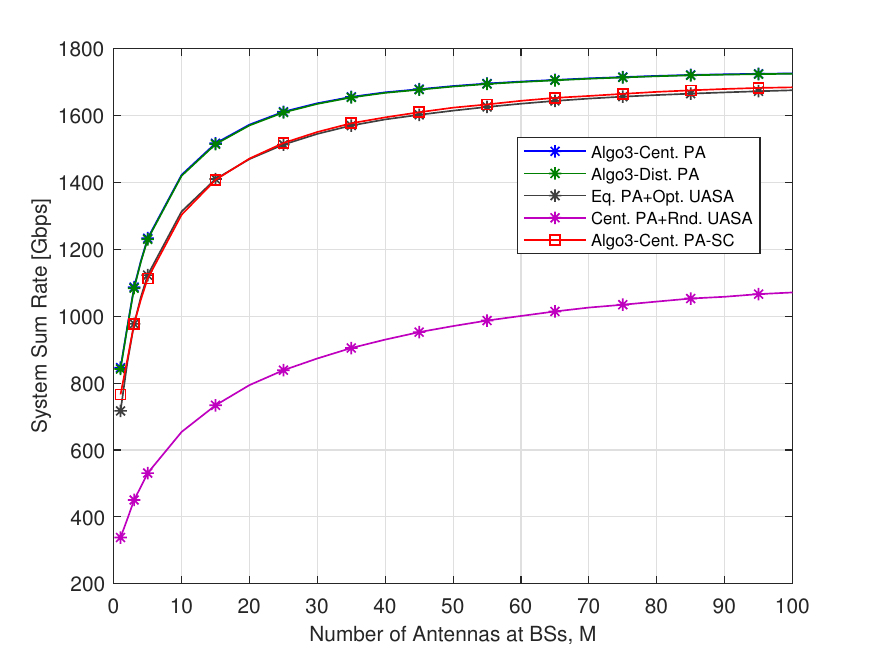}
        \caption{{System sum-rate vs the number of antennas at BSs, $q=0.5$, $B=10$, $N=12$}}
        \label{fig:NumAntsVaries}
    \end{minipage}\hfill
    \begin{minipage}{0.32\textwidth}
        \centering
        \includegraphics[scale=0.43]{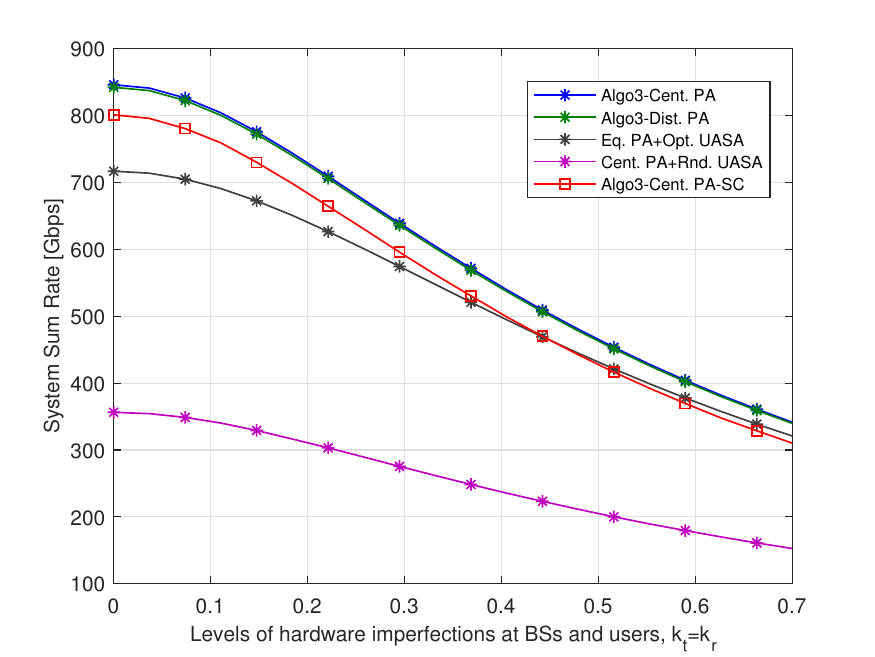}
        \caption{{System sum-rate vs the level of HI at BSs and users ($k_t=k_r$), $q=0.5$, $B=10$, $N=12$}}
        \label{fig:HWI_Varies}
    \end{minipage}\hfill
    \begin{minipage}{0.32\textwidth}
        \centering
        \includegraphics[scale=0.43]{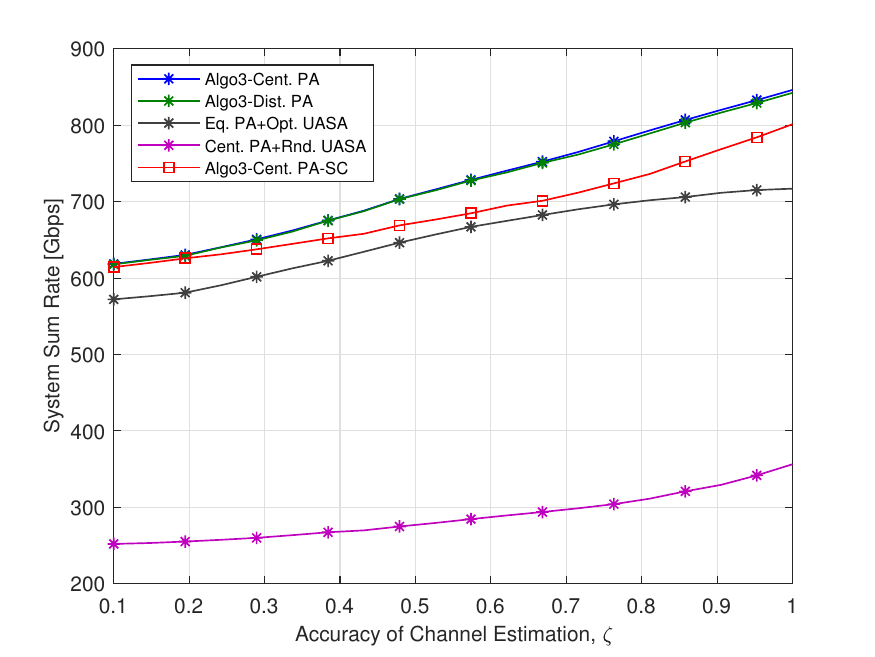}
        \caption{{System sum-rate vs the accuracy of channel estimation ($\zeta$), $q=0.5$, $B=10$, $N=12$}}
        \label{fig:ChanEstAccu_Varies}
    \end{minipage}
\end{figure*}
%%%%%%%%%%%%%%%

{Fig.~\ref{fig:DifTWs-BW-SubBands-MolecAbs}a demonstrates the impact of varying the tolerance of variation within each TW, i.e., $\epsilon$, on the total available bandwidth, which is obtained as $\bar{w}_T = f_e-w_e -(f_I+w_I)$, and the optimal number of sub-bands, achieved in \eqref{eqn:LowerBound-S}. For $\text{TW}_1$, the total available bandwidth is less than that of $\text{TW}_3$ since the length of the TW is smaller. Moreover, the smaller TW length of $\text{TW}_1$ leads to having higher edge bands compared to $\text{TW}_3$ for a given $\epsilon$, leading to a reduction in the available bandwidth. Although increasing the tolerance $\epsilon$ leads to higher available bandwidth, it reduces the edge bands and therefore causes a higher molecular absorption coefficient as shown in Fig.~\ref{fig:DifTWs-BW-SubBands-MolecAbs}b.}

\subsubsection{{Multi-Antenna Base Stations}}

{In this sub-section, we aim to show the performance of the proposed resource allocation methods under the scenario where each BS is equipped with a uniform linear array comprising $M$ antenna elements. For beamforming, maximum ratio transmission (MRT) beamforming is utilized. Denoting the beamformer at $b$-th BSs on sub-band $s$ by $\boldsymbol{u}_{b,s}=\boldsymbol{h}_{b,s,n}\in\mathbb{C}^{M\times 1}$ \cite{MRT-1}, the SINR at user $n$ from BS $b$ on sub-band $s$ is given by:}
{
\begin{equation}\label{eqn:SINR-MultiAnt}
    \gamma_{b,s,n} = \frac{p_{b,s}\abss{\boldsymbol{h}_{b,s,n}^H \boldsymbol{u}_{b,s}} }{\Bpsum q p_{\bp,s}\abss{\bar{\boldsymbol{h}}_{\bp,s,n}^H \boldsymbol{u}_{\bp,s}} + p_{b,s}\abss{\tilde{\boldsymbol{h}}_{b,s,n}^H \boldsymbol{u}_{b,s}} + {N_0} w},
\end{equation}
where the notation $\boldsymbol{h}\in\mathbb{C}^{M\times 1}$ in boldface represents the equivalent complex vector of their single-antenna channel counterpart.
Fig.~\ref{fig:NumAntsVaries} illustrates the effect of varying the number of antennas at each BS. The results indicate that increasing $M$ enhances the proposed solutions, surpassing other methods. Moreover, as the number of antennas increases, the performance of ``Algo3-Cent. PA-SC" becomes comparable to that of ``Eq. PA+Opt. UASA," which benefits from multi-connectivity.
}
\subsubsection{{Impact of Hardware Imperfections}}

{ In this sub-section, we account for the influence of hardware imperfections (HI) in THz transceivers on the system's performance. Consequently, we incorporate the impact of HI into the SINR formulation. Distortions stemming from HI at both transmitters and receivers can be modeled as additive Gaussian noise \cite{HWI-2,HWI-1}. For the received signal at a user $y=h(x+n_t)+n_r$, $n_t \sim \mathcal{CN}(0,k_t^2 p)$ represents the distortion noise of the transmitter, and $n_t \sim \mathcal{CN}(0,k_r^2 \abss{h}p)$ denotes the distortion at the receiver side. Here, $k_t$ and $k_r$ are the HI parameters at the transmitter and receiver, respectively. Note that the variance of the distortion noise is proportional to the input power of the hardware. Employing this HI modeling, the HI-aware SINR at user $n$ associated with BS $b$ over sub-band $s$ is formulated as follows:
\begin{equation}\label{eqn:HI-SINR}
    \gamma_{b,s,n}^{\mathrm{HI}} 
     = \frac{p_{b,s} \h}{\bar{I}_{b,s,n}(\p) \! + \! K_{b,s,n}(\p) + p_{b,s}\htilde \!\!\!+\!\! N_0 w},
\end{equation}
where $\bar{I}_{b,s,n}(\p)$ is defined in \eqref{eqn:CumulInter}, and 
\begin{equation}\label{eqn:HI-Power}
    K_{b,s,n}(\p)=(k_t^2 + k_r^2)p_{b,s}\h + k_r^2 (\Bpsum q p_{\bp,s}\abss{{h}_{\bp,s,n}}).
\end{equation}
Based on $\gamma_{b,s,n}^{\mathrm{HI}}$, all of the formulations used in Algorithms 1,2, and 3 can be modified. 
The impact of different levels of hardware imperfections on the performance of the proposed solutions is depicted in Fig.~\ref{fig:HWI_Varies}. It is observed that the proposed algorithms outperform other benchmarks while exhibiting a decreasing overall performance as $k_r=k_t$ increases. Additionally, the performance of the single-connectivity case deteriorates more rapidly, approaching the case of ``Eq. PA+Opt. UASA" at higher levels of HI.
}
\subsubsection{{Impact of Imperfect CSI}}

{For the case of imperfect CSI, we consider the following model $
    h = \zeta \hat{h} + \sqrt{1-\zeta^2}e,
$
where $\hat{h}$ is the channel estimation of $h$, $\zeta$ is the accuracy of channel estimation, and $e$ is the estimation error that follows a complex Gaussian distribution with zero mean and variance of $d^{-2}$, such that $d$ is the distance of the considered link \cite{Imp-CSI-1}. Fig.~\ref{fig:ChanEstAccu_Varies} demonstrates the impact of imperfect CSI on the proposed algorithms, where increasing $\zeta$ leads to approaching perfect CSI.  The results show that the performance of both the proposed solutions and other benchmarks improves with higher channel estimation accuracy. However, the improvement of ``Eq. PA+Opt. UASA" is less pronounced compared to other approaches due to its lack of optimal power allocation.
}

\section{Conclusion}
This paper considered molecular absorption-aware and blockage-aware resource allocation within a multi-cell THz system, considering multi-connectivity to maximize the system's total sum-rate. To make the problem tractable, we first proposed a convex approximation of the molecular absorption coefficient and derived a closed-form lower bound for the optimal number of sub-bands with consideration to beam-squint effects. The optimization problem is then decomposed into the power allocation, and the joint user association and sub-band assignment sub-problems. Leveraging the concepts from unimodular matrices, various low-complexity centralized and distributed solutions are presented.
% To address the latter, a solution of low complexity is proposed, capitalizing on the concept of constraint matrix unimodularity. This approach facilitates the relaxation of binary variables without compromising optimality. In the context of power allocation, the study adopts FP to propose a sub-optimal centralized solution. Additionally, by integrating the ADMM with FP, a distributed low-complexity power allocation approach is introduced. The overall proposed algorithm alternatively solves the power allocation and the joint user association and sub-band assignment problems until convergence is achieved. 
Numerical findings underscore the proposed solutions' superiority compared to traditional benchmarks, {as well as the advantages of multi-connectivity compared to single-connectivity.} Also, the impact of varying blockage density along with the antennas' directionality, {hardware impairment, and imperfect CSI are explored. Moreover, the proposed framework can be extended to include explicit minimum rate requirements for users by using SCA, and FP, as in \cite{QoS-1,QoS-2,QoS-3}, to enhance performance from the users' perspective.} 
% Also, the influence of increased numbers of BSs is investigated on multi-connectivity as well as the size of the problem that can be tackled by the proposed low-complex solutions. 

% \appendix

\section*{Appendix}
\subsection{Proof of \textbf{Lemma 3}}
\label{FirstAppendix}

\noindent To prove $\mathbf{T}^{(3)}$ is TUM, we utilize the following Theorem:
\begin{theorem} 
A $m \times n$ matrix $A$ is TUM \textbf{iff}  $\forall I \subseteq \{1,\ldots,m\}$, there exists a partition of $I$ into sets $F$ and $L$ such that for all $j \in \{1,\ldots,n\}$,  $\left| \sum_{i \in F} A_{ij} - \sum_{i \in L} A_{ij} \right| \leq 1$ holds \cite{schrijver1998theory}.
\label{Theo_3}
\end{theorem}
To apply \textbf{Theorem~\ref{Theo_3}}, we need to show that for any subset of rows of $\mathbf{T}^{(3)}$, we can find a partitioning of the rows into two groups ($G_A$ and $G_B$) where for each column, the sum of the elements in group $A$ minus the sum of the elements in group $B$ is either 0, 1, or -1. To do this, we first divide rows of $\mathbf{T}^{(3)}$ into four groups denoted by $\textit{\textrm{R}}_1, \textit{\textrm{R}}_2, \textit{\textrm{R}}_3$ and $\textit{\textrm{R}}_4$. The columns are also divided into B groups. In Figure \ref{Unimodular_Mat}, a detailed presentation of $\mathbf{T}^{(3)}$ is provided where we can see details of the first matrices in $\textit{\textrm{C}}_1$.

% This is commented for the sake of reducing compiling time:

Now, let's assume that $K$ rows are selected, i.e. $I = \{r_1, r_2, \dots, r_K\}$. Considering the groups of rows, we divide $I$ into four sets $I_1, I_2, I_3,$ and $I_4$, where $I_j \subseteq \textit{\textrm{R}}_j$ and $\sum_{j = 1}^{4} |I_j| = K$. Based on this, we propose the following partitioning of the rows. All the rows in $I_2$ and $I_3$ go to group $G_B$, and all the rows in $I_4$ go to group $G_A$. Rows in $I_1$ go to either group $G_A$ or $G_B$ based on the following rule. Consider $r_n \in I_1$, we have:
% \begin{align}
% \forall r^1_n \in I_1 \quad if \quad \exists r^2_m \in I_2, n &\equiv m \pmod{S} \longrightarrow r^1_m \in G_B, \quad otherwise \quad r^1_m \in G_A
% \end{align}
\begin{equation}
    \text{Group of } r_n = 
\left\{
	\begin{array}{ll}
		G_B  & \text{if} \quad \exists r_m \in I_2, \quad n\equiv m \pmod{S} \\
		  G_A & \mbox{otherwise}
	\end{array}
\right.
\end{equation}
This condition comes from two facts about $\mathbf{T}_3$. First, in each column, we have exactly 3 non-zero elements. Second, if we index the rows in $\textit{\textrm{R}}_1$ from 1 to $NS$, and the rows in $\textit{\textrm{R}}_2$ from 1 to $S$, we can see that for each column the rows of the first two non-zero elements belong to $\textit{\textrm{R}}_1$ and $\textit{\textrm{R}}_2$, and their row index have the same remainder with respect to $S$. 

Now, we show that the \textit{proposed partitioning follows the conditions of} \textbf{Theorem~\ref{Theo_3}}.
Consider a column from $\textit{\textrm{C}}_1$, $c$. Since the maximum number of non-zero elements in each column is three, we'll face one of the following scenarios:

$\bullet$ There isn't any non-zero element in $c$, which means that $|\sum_{i \in G_A} \mathbf{T}^{(3)}_{i, c} -  \sum_{i \in G_B} \mathbf{T}^{(3)}_{i, c}| = 0$

$\bullet$ There is only one 1 or -1 present in $c$, i.e., the sum of the elements in one of the groups ($G_A$ or $G_B$) is 0 and the other one is -1 or +1, which makes their difference -1 or +1. 

$\bullet$ There are two non-zero elements in $c$. Let's denote the rows corresponding to the non-zero elements by $r_n$, $r_m$ ($n < m$). Since all elements in $\textit{\textrm{R}}_3$ are zero, we have three different scenarios. First, $r_n \in \textit{\textrm{R}}_1$ and $r_m \in \textit{\textrm{R}}_2$. In this case, they both are in group $G_B$, and since all other elements are zero, we have $|\sum_{i \in G_A} \mathbf{T}^{(3)}_{i, c} -  \sum_{i \in G_B} \mathbf{T}^{(3)}_{i, c}| = |0 - \mathbf{T}^{(3)}_{n, c} - \mathbf{T}^{(3)}_{m, c}|=|0 - (-1) - (+1)| = 0$. Second,  $r_n \in \textit{\textrm{R}}_2$ and $r_m \in \textit{\textrm{R}}_4$. Since all elements of $\textit{\textrm{R}}_2$ are in $G_B$, and all  elements of $\textit{\textrm{R}}_4$ are in $G_A$, we have:  $|\sum_{i \in G_A} \mathbf{T}^{(3)}_{i, c} -  \sum_{i \in G_B} \mathbf{T}^{(3)}_{i, c}| = |\mathbf{T}^{(3)}_{n, c} - \mathbf{T}^{(3)}_{m, c}|=|(+1) - (+1)| = 0$. Third, $r_n \in \textit{\textrm{R}}_1$ and $r_m \in \textit{\textrm{R}}_4$. In this case, since we don't have any non-zero element from $\textit{\textrm{R}}_2$, this means that no row from $\textit{\textrm{R}}_2$ is selected that has the same remainder with respect to $S$ as $n$. Thus, $r_n \in G_A$. This means  $|\sum_{i \in G_A} \mathbf{T}^{(3)}_{i, c} -  \sum_{i \in G_B} \mathbf{T}^{(3)}_{i, c}| = |\mathbf{T}^{(3)}_{n, c} + \mathbf{T}^{(3)}_{m, c} - 0|=|(-1) + (+1) - 0| = 0$. 
    
$\bullet$ There are three non-zero elements in $c$. Let's denote the rows with $r_n$, $r_m$, and $r_k$, where $n < m< k$. Based on the structure of $\mathbf{T}^{(3)}$, we know that $r_n \in \textit{\textrm{R}}_1$, $r_m \in \textit{\textrm{R}}_2$, and $r_k \in \textit{\textrm{R}}_4$. Thus, $r_n, r_m \in G_B$ and $r_k \in G_A$. Thus, we have $|\sum_{i \in G_A} \mathbf{T}^{(3)}_{i, c} -  \sum_{i \in G_B} \mathbf{T}^{(3)}_{i, c}| = |\mathbf{T}^{(3)}_{k, c} - \mathbf{T}^{(3)}_{n, c} - \mathbf{T}^{(3)}_{m, c}|=|(+1) - (-1) - (+1)| = +1$. 

Therefore, the proposed partitioning follows the requirements of \textbf{Theorem~\ref{Theo_3}} 
% \textbf{(1)}~there isn't any non-zero element in $c$, which means that $|\sum_{i \in G_A} T^{(3)}_{i, c} -  \sum_{i \in G_B} T^{(3)}_{i, c}| = 0$, \textbf{(2)} there is only one 1 or -1 present in $c$.  This means that sum of the elements in one of the groups ($G_A$ or $G_B$) is zero and the other one is -1 or +1, which makes their difference -1 or +1. Thus, the partitioning follows the condition of \textbf{Theorem~\ref{Theo_3}}, \textbf{(3)}~there are two non-zero elements in $c$. Thus, the partitioning follows the condition of \textbf{Theorem~\ref{Theo_3}.}
in $\textit{\textrm{C}}_1$. In a similar manner, one can show that it's also true for other columns $c \in \textit{\textrm{C}}_k$, where $2 \leq k \leq \textrm{B}$. Since each column has at most three non-zero elements, we have four similar scenarios as discussed above. 
Subsequently, for any column $c$ of $\mathbf{T}^{(3)}$, i.e. $c \in \textit{\textrm{C}}_k$, where $k \in \{1, 2, \dots \textrm{B}\}$, the proposed partitioning of the selected rows into $G_A$ and $G_B$ follows the requirement of Theorem \ref{Theo_3}. This means that $\mathbf{T}^{(3)}$ is TUM, implying that $\mathbf{T}$ is TUM; thus, proving Lemma \ref{Lemma_2}.

\subsection{{Proof of Convergence for \textbf{Algorithm-3} }}
\label{SecondAppendix}
{
For the first stage, based on \textbf{Lemma 1}, we showed that the lower bound on $S$ is optimal. In step 3 of the second stage at the $i$-th iteration, for fixed power allocation $\boldsymbol{\mathbf{P}}^{(i-1)}$, the globally optimal values of user association and sub-band assignment variables, $\A^{(i)}$, are obtained due to the convexity of the transformed linear optimization problem in $\mathcal{P}_3$.} 

{
Then, for fixed $\A^{(i)}$, we need to show that solving $\mathcal{P}_4$ converges to a stationary point for $\p^{(i)}$. Denoting the objective function of $\mathcal{P}_4$ at the $t$-th iteration of either Algorithm 1 or 2 by $f(\p^{(t)})$, for fixed $\A^{(i)}$, we have:
\begin{align}
    f(\p^{(t)}) &=  f_1(\p^{(t)},\boldsymbol{\gamma}^{(t)}) \label{OptProof-1} \\
    & \geq f_1(\p^{(t)},\boldsymbol{\gamma}^{(t-1)}) \label{OptProof-2} \\
    & = f_2(\p^{(t)},\boldsymbol{\gamma}^{(t-1)}, \boldsymbol{\bar{Y}}^{(t-1)}) \label{OptProof-3} \\
    & \geq f_2(\p^{(t)},\boldsymbol{\gamma}^{(t-1)}, \boldsymbol{Y}^{(t-1)}) \label{OptProof-4} \\
    & \geq f_2(\p^{(t-1)},\boldsymbol{\gamma}^{(t-1)}, \boldsymbol{Y}^{(t-1)}) \label{OptProof-5} \\
    &=  f_1(\p^{(t-1)},\boldsymbol{\gamma}^{(t-1)}) \label{OptProof-6} \\
    &=  f(\p^{(t-1)}) \label{OptProof-7}
\end{align}
where \eqref{OptProof-1} is due to the fact that by substituting optimal $\boldsymbol{\gamma}$ in $f_1$, the original objective function can be obtained. \eqref{OptProof-2} holds since updating $\gamma$ maximizes $f_1$ when other variables are fixed. \eqref{OptProof-3} follows \textbf{Remark 1} such that $\boldsymbol{\bar{Y}}$ is obtained using equation \eqref{eqn:Optimal-Y}. \eqref{OptProof-4} holds since updating $\boldsymbol{Y}$ maximizes $f_2$ while other variables are fixed. \eqref{OptProof-5} follows the fact that updating $\p$ maximizes $f_2$ when other variables are fixed. \eqref{OptProof-6} and \eqref{OptProof-7} follows the similar reasoning as \eqref{OptProof-3} and \eqref{OptProof-1}, respectively. Therefore, $f$ is monotonically non-decreasing after each iteration. Moreover, since the objective function, $f$ is bounded due to the constraints, Algorithms 1 and 2 converge. At the convergence, a local optimal point of the reformulated objective function $f_2$ is obtained. The solution is a stationary point of the objective function $f$ in $\pr_4$ \cite{FP-Part2}. Therefore, since the original objective function is non-decreasing after each sub-problem, the alternating optimization in stage 2 of Algorithm 3 converges to a sub-optimal solution.}

\bibliography{ref}
\bibliographystyle{IEEEtran}

\begin{IEEEbiography}[{\includegraphics[width=1in,height=1.25in, clip,keepaspectratio]{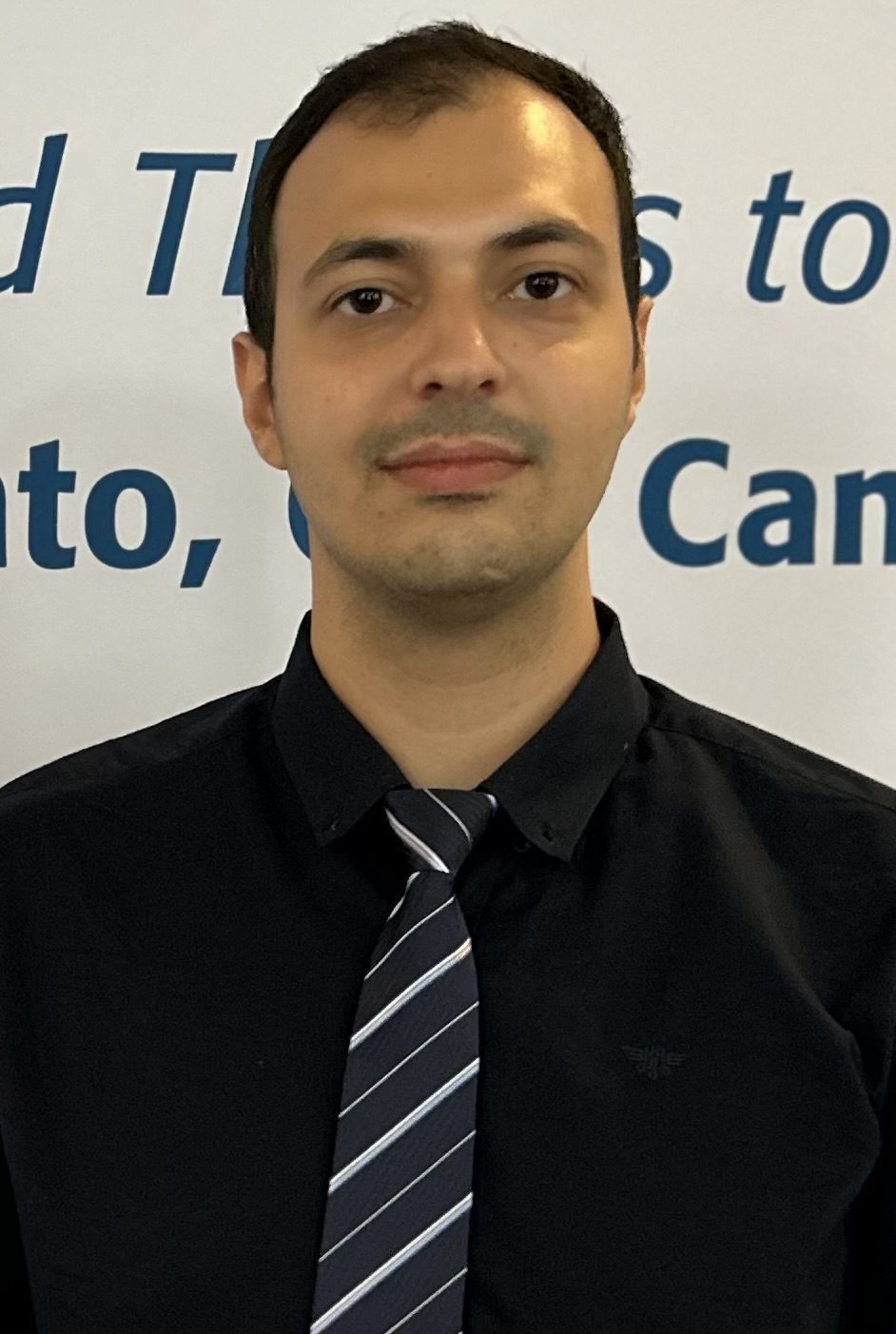}}]{\textbf{Mohammad Amin Saeidi}} received the M.Sc. degree in Electrical Engineering -- Communication Systems from the Amirkabir University of Technology, Tehran, Iran, in 2021. Currently, he is pursuing a Ph.D. degree in Electrical Engineering and Computer Science at York University, Canada. His research focuses on topics, including, resource management and optimization in wireless communications, terahertz communication, multi-band networks, and intelligent reflecting surfaces. He has served as a reviewer in various IEEE journals, including  IEEE TCOM, IEEE TWC, IEEE TMC, etc.
\end{IEEEbiography}

\begin{IEEEbiography}[{\includegraphics[width=1in,height=1.25in, clip,keepaspectratio]{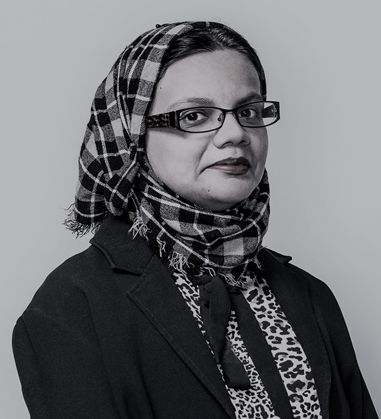}}]{\textbf{Hina Tabassum }} \hspace{0.01em} (Senior Member, IEEE) received the Ph.D. degree from the King Abdullah University of Science and Technology (KAUST). She is currently an Associate Professor with the Lassonde School of Engineering, York University, Canada, where she joined as an Assistant Professor, in 2018. She is also appointed as the York Research Chair of 5G/6G-enabled mobility and sensing applications (2023 - 2028). She was a postdoctoral research associate at University of Manitoba, Canada.  She has published over 100 refereed papers in well-reputed IEEE journals, magazines, and conferences. Her research interests include multi-band optical, mm-wave, and THz networks and cutting-edge machine learning solutions for next generation wireless communication and sensing networks. She received the Lassonde Innovation Early-Career Researcher Award in 2023 and the N2Women: Rising Stars in Computer Networking and Communications in 2022. She was listed in the Stanford’s list of the World’s Top $2\%$ Researchers in 2021, 2022, and 2023.  She is the Founding Chair of the Special Interest Group on THz communications in IEEE Communications Society (ComSoc)-Radio Communications Committee (RCC). She served as an Associate Editor for IEEE Communications Letters (2019–2023), IEEE Open Journal of the Communications Society (OJCOMS) (2019–2023), and IEEE Transactions on Green Communications and Networking (TGCN) (2020–2023). Currently, she is also serving as an Area Editor for IEEE OJCOMS and an Associate Editor for IEEE Transactions on Communications, IEEE Transactions on Wireless Communications, and IEEE Communications Surveys and Tutorials. She has been recognized as an Exemplary Editor by the IEEE Communications Letters (2020), IEEE OJCOMS (2023), and IEEE TGCN (2023).  
\end{IEEEbiography}

\begin{IEEEbiography}[{\includegraphics[width=1in,height=1.25in, clip,keepaspectratio]{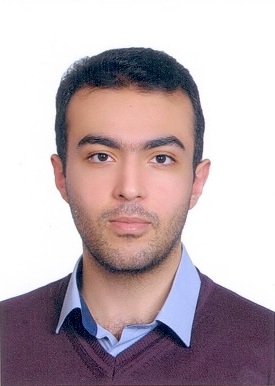}}]{\textbf{Mehrazin Alizadeh}} obtained his M.Sc. degree in Computer Engineering from York University, Toronto, ON, Canada, in 2022, after receiving a B.Sc. degree in Electrical Engineering from the University of Tehran, Tehran, Iran, in 2020. His research is fundamentally centered around the applications of deep learning in addressing resource allocation problems in wireless communications networks. In particular, his work involves leveraging deep neural networks to design novel solvers for optimization problems within wireless networks, offering significant computational advantages over traditional methodologies.
\end{IEEEbiography}

\end{document}